\documentclass{acm_proc_article-sp}

\usepackage{amsmath,amsfonts,amssymb}
\usepackage[ruled,vlined,linesnumbered]{algorithm2e}
\usepackage{comment}
\usepackage[all]{xy}
\usepackage{overpic}
\usepackage{epic,eepic}

\newcommand\numberthis{\addtocounter{equation}{1}\tag{\theequation}}
\newcommand{\yy}{\mathbf{y}}
\newcommand{\xx}{\mathbf{x}}

\newcommand{\RR}{\mathbb{R}}
\newcommand{\hh}{\mathcal{H}}

\newcommand{\oomit}[1]{}
\newcommand{\define}{\widehat{=}}

\newcommand{\fb}{\mathbf{f}}

\def \CALC {{\mathcal C}}
\def \CALX {{\Xi}}
\def \LIDER  {{\mathcal L}}
\def \II  {{\CALX_0}}
\def \UU {S^u}
\def \RE  {{\mathcal{R}}}
\def \RH {{\RE_{\hh}}}
\def \psd  {{\tt CADpsd}}
\def \sdp {{\bf SDP}}

\def \ec {{\it exponential condition}}
\def \bc {{\it barrier certificate}}

\newtheorem{theorem}{Theorem}
\newtheorem{problem}{Problem}
\newtheorem{lemma}{Lemma}
\newtheorem{example}{Example}
\newtheorem{definition}{Definition}
\newtheorem{remark}{Remark}
\newtheorem{corollary}{Corollary}

\begin{document}

\title{Barrier Certificates Revisited}

\numberofauthors{4}

\author{
%
%
	\alignauthor \hspace*{-2cm}
	\mbox{Liyun Dai, Ting Gan and Bican Xia} \\
	\hspace*{-2cm} \affaddr{\mbox{LMAM \& School of Mathematical Sciences,} \\
   \hspace*{-5.8cm} Peking University }\\
	  \email{\hspace*{-3cm} \{dailiyun,gant\}@pku.edu.cn}\\
      \email{\hspace*{-3cm} xbc@math.pku.edu}
	\and  
	 \alignauthor Naijun Zhan\\
	\affaddr{\mbox{State Key Lab. of Computer Science,} \\
  \mbox{Institute of Software, Chinese Academy of Sciences}}\\
	\email{znj@ios.ac.cn}
}

\maketitle
\begin{abstract}
	A barrier certificate can separate the state space of a considered hybrid system (HS)
  into safe and unsafe parts according to the safety property to be verified. Therefore this
  notion has been widely used in the verification of HSs.
    A stronger condition on 
    barrier certificates means that less expressive barrier certificates can be synthesized.
    On the other hand, synthesizing more expressive barrier certificates often means
    high complexity.
     In \cite{kong2013}, Kong et al considered how to
     relax the condition of barrier certificates while still keeping
      their convexity so that one can synthesize more expressive barrier certificates
       efficiently using semi-definite programming ({\sdp}).
     In this paper, we first discuss how to relax
     the condition of barrier certificates in a general way, while still keeping
      their convexity. Particularly, one can then utilize different weaker conditions flexibly to
       synthesize different kinds of barrier certificates with more expressiveness
         efficiently using {\sdp}.
       These barriers give more opportunities to verify the considered system.
        We also show how to combine two functions together to form a combined barrier certificate
         in order to prove a safety property under consideration, whereas neither of them can be used as a barrier certificate separately,
        even according to any relaxed condition.  Another contribution of this paper is that we
      discuss how to discover certificates
     from the general relaxed condition by {\sdp}. In particular, we
     focus on how to avoid the unsoundness because of numeric error caused by {\sdp}
     with symbolic checking.
\end{abstract}

\category{F.3.1}{Specifying and Verifying and Reasoning about Programs}{Invariants}

\terms{Theory}

\keywords{Inductive invariant, barrier certificate, safety verification,
hybrid system, nonlinear system, sum of squares} 

\section{Introduction}\label{sec:inr}

Embedded systems make use of computer units to control physical devices so that the behavior of the controlled devices meets expected requirements. They  have become ubiquitous in our modern life.
 How to design correct embedded systems is a grand challenge for computer science and control theory.
Model-driven development (MDD) is considered as an effective way of developing correct complex embedded systems, and
has been successfully applied in industry \cite{HS06,Lee00}. In the framework of MDD, a formal model of the system to be developed is defined at the
beginning; then extensive analysis and verification are conducted based on the formal model so that errors can be detected and corrected at the very
early stage of the design of the system. Afterwards, model transformation techniques are applied to transform the abstract formal model into lower
level models, even into source code. Hybrid systems (HSs) combine discrete mode changes with continuous evolutions specified in the form of differential
equations. With mathematically precise semantics, HSs can serve as an appropriate model of embedded systems \cite{MMP92,ACHH93}.

In the past, analysis and verification of HSs are mainly done through directly computing reachable sets, either by model-checking (e.g.,
\cite{ACH95,PV94,HH95}) or by
 decision procedures (e.g., \cite{LPY01}). The basic idea is to partition the state space of a considered system into
  finite many equivalent classes, or represent to finite many computable sets according to
  the solutions of the ODEs of the system. Since there is only a very small class of ODEs with closed form solutions, the scalability of these approaches is very restricted, only applicable to very specific linear HSs.
To deal with more complicated systems, a deductive method has been recently proposed and successfully applied in practice \cite{PlatzerClarke08,PlatzerClarke09}.
The most challenging part of a deductive method is how to discover invariants,
which hold at all reachable states of the system. For technical reason,
people only consider how to synthesize inductive invariants, which are preserved by all discrete and continuous transitions.
In general, a safety property itself is an invariant, but not an inductive invariant.
Obviously, an inductive invariant is an approximation of the reachable set, which may be discovered
according to the ODEs, rather than their solutions.
The basic idea is as follows: first, predefine a property template (linear or non-linear, depending on the property to be verified);
 then, encode the conditions of a property to be inductive (discretely and/or continuously) into some constraints on
 state variables and parameters;
 finally, find out solutions to the constraints.
 So, how to define inductiveness conditions and the power of constraint solving are essential in these
 approaches.

 Many approaches have been proposed  following the line discussed above. E.g.,
 in \cite{JM98,rodriguez2005}, the authors independently proposed different approaches
 for constructing inductive invariants for linear HSs; S. Sankaranarayanan et al presented a computational method to automatically
 generate algebraic invariants for algebraic HSs in \cite{sankaranarayanan2010,SSM04a},
 based on the theory of pseudo-ideal over polynomial ring and quantifier elimination;
S. Prajna in \cite{PJ04,prajna2007} provided a new notion of
inductive invariants called \emph{barrier certificates} for verifying the safety of semi-algebraic HSs with stochastic setting using the technique of sum-of-squares (\textbf{SOS});
while in \cite{PlatzerClarke08}, Platzer and Clarke extended the idea of \emph{barrier certificates} \ by considering
  boolean combinations of multiple polynomial inequalities;
   In \cite{GT08,TT09},  S. Gulwani et al investigated how to generate
   inductive invariants with more expressiveness for semi-algebraic HSs by relaxing the inductiveness conditions
   by considering inductiveness on the boundaries of predefined invariant templates; while in \cite{LZZ11}, Liu et al considered how to further relax
   the inductiveness condition given in \cite{GT08,TT09} and first gave a complete method on how to
   generate semi-algebraic invariants for semi-algebraic HSs.
  In \cite{sloth2012}, C. Solth at el proposed an approach to constructing global inductive invariant
   from local differential invariants using optimization technique.

   The aforementioned approaches can be classified into two categories: symbolic computation based approaches like
    \cite{JM98,rodriguez2005,sankaranarayanan2010,SSM04a,PlatzerClarke08,GT08,TT09,LZZ11}, and numeric computation
    based approaches like \cite{PJ04,prajna2007,sloth2012}. In general, the former can
    synthesize more expressive invariants, but their efficiencies are very low; in contrast, the efficiency of
    the latter is very high, normally in  \emph{polynomial time}
      as only {\sdp} is used, but the expressiveness of
      synthesized invariants is restrictive.
     In \cite{kong2013}, Kong et al investigated how to synthesize more expressive barrier certificates
      by proposing \emph{exponential barrier certificate condition}, which is
       a relaxed inductiveness condition, but still keeps the convexity of {\bc}s. Therefore, more expressive
       barrier certificates can be synthesized efficiently according to their condition still by {\sdp}.

In this paper, firstly, following Kong et al's line,
	 in the prerequisite of keeping the convexity of barrier certificates so that
     {\sdp} is still applicable, we  discuss how to relax
     the condition of barrier certificates in a general way.
       Thus, one can utilize different weaker conditions flexibly to
       synthesize different kinds of barrier certificates with more expressiveness efficiently,
       which gives more opportunities to verify the considered system.
       In addition, we consider how to combine two functions together to form a combined barrier certificate to prove a safety property under consideration, whereas  neither of these two functions can be used as a barrier certificate separately,
        even according to any relaxed conditions.  Another contribution of this paper is that we
      design algorithms to synthesize barrier certificates
     according to the general relaxed condition by {\sdp}. In particular, we
     focus on how to avoid the unsoundness of our approach caused by numerical errors in {\sdp}.

 The rest of the paper is organized as follows: Section~\ref{sec:pre}  introduces some basic notions;
 In Section~\ref{sec:bccondtion}, we discuss how to relax barrier certificate conditions, as well
  how to combine two functions to form a combined barrier certificate, but neither of them can be used as
   a barrier certificate separately;  Section~\ref{sec:alg} is devoted to
    how to synthesize barrier certificates according to relaxed conditions discussed above based on {\sdp};
    Section~\ref{sec:exp} provides some case studies as well as experimental results.
    Finally, we conclude this paper in Section~\ref{sec:con}.

\section{Preliminaries} \label{sec:pre}
In this section, we first introduce some basic notions, and then explain the basic idea of barrier certificates.

In what follows, we use $\RR$ to stand for the set of reals, $\CALC^\omega[\RR^n]$ for the set of analytic function
     from $\RR^n$ to $\RR$.

\subsection{Basic notions}
An autonomous continuous dynamical system (CDS)  is represented by a differential equation of the form
\begin{equation}
	\dot{\xx}=\fb(\xx)
	\label{eq:pro}
\end{equation}
where $\xx\in \RR^n$, and $\fb$ is a  vector function, called \emph{field vector},
   whose components are in $\CALC^\omega[\RR^n]$,
 and satisfy local Lipschitz condition\footnotemark. \footnotetext{Local Lipschitz condition
  guarantees the
existence and uniqueness of the solution of (\ref{eq:pro}) from any initial $\xx_0$.}
In the context of HSs, a CDS is normally equipped with a domain $D\subseteq \RR^n$ defining its state space
  and an initial set
  of states $\Xi$.

In this paper, we use hybrid automata  \cite{ACH95}  to model HSs,
 more models of HSs can be found in
\cite{ZWZ13}.

\begin{definition}[Hybrid Automata]\label{dfn:HS}\,
A {hybrid automaton} (HA) is a system $\hh\,\define\,(Q,X,f,D,E,G,R,\Xi)$,\, where \vspace*{-4mm}
\begin{itemize}
  \item $Q=\{q_1,\ldots,q_m\}$ is a finite set of discrete states (or modes); \vspace*{-2mm}
  \item $X=\{x_1,\ldots,x_n\}$ is a finite set of continuous state variables, with $\xx=(x_1,\ldots,x_n)$ ranging over $\RR^n$; \vspace*{-2mm}
  \item $f: Q\rightarrow (\RR^n\rightarrow \RR^n)$ assigns to each mode $q\in Q$ a locally Lipschitz continuous vector field $\fb_q$; \vspace*{-2mm}
  \item $D$ assigns to each mode $q\in Q$ a mode domain $D_q\subseteq \RR^n$; \vspace*{-2mm}
  \item $E\subseteq Q\times Q$ is a finite set of discrete transitions; \vspace*{-2mm}
  \item $G$ assigns to each transition $e\in E$ a switching guard $G_e\subseteq \RR^n$; \vspace*{-2mm}
  \item $R$ assigns to each  transition $e\in E$ a reset function $R_e$: $\RR^n \rightarrow {\RR^n}$; \vspace*{-2mm}
  \item $\Xi$ assigns to each $q\in Q$ a set of initial states $\Xi_q\subseteq \RR^n$. \vspace*{-5mm}
\end{itemize}
\end{definition}

For ease of presentation, we make the following assumptions: \vspace*{-4mm}
  \begin{itemize}
  \item for all $q\in Q$, $\fb_q$ is a polynomial vector function, so it  satisfies local Lipschitz condition, and thus the existence and uniqueness of solutions to $\dot \xx=\fb_q$ is guaranteed; \vspace*{-2mm}
  \item  for all $q\in Q$ and all $e\in E$, $\Xi_q$ is a \emph{semi-algebraic} set, $D_q$ and $G_e$ are \emph{closed semi-algebraic} sets\footnote{A subset $A\subseteq \RR^n$ is called \emph{semi-algebraic} if there is a quantifier-free polynomial formula $\varphi$ expressed in Tarski's algebra s.t.
$A=\{\xx\in\RR^n\mid\varphi(\xx)\,\,\mbox{is true}\}$\,.}. \vspace*{-2mm}
\end{itemize}

Given an HA $\hh$, a \emph{safety requirement} $S$ of $\hh$ assigns to each mode $q\in Q$
a safe region $S_q\subseteq \RR^n$, i.e. $S=\bigcup_{q\in
Q}(\{q\}\times  S_q)$. Dually, $\UU=\bigcup_{q\in
Q}(\{q\}\times (D_q-S_q))$ is called \emph{unsafe set}.
The \emph{reachable set} of $\hh$, denoted by $\RH$, consists of those $(q,\xx)$ for which there exists a finite sequence
$$(q_0,\xx_0),(q_1,\xx_1),\ldots,(q_l,\xx_l)$$
s.t. $(q_0,\xx_0)\in \Xi_{\hh}$, $(q_l,\xx_l)=(q,\xx)$, and for any $0\leq i\leq l-1$, one of the following two conditions holds: \vspace*{-4mm}
\begin{itemize}
  \item (Discrete Jump): $e=(q_i,q_{i+1})\in E$, \,$\xx_i\in G_e$ and $\xx_{i+1}=R_e(\xx_i)$; or \vspace*{-2mm}
  \item (Continuous Evolution): $q_i=q_{i+1}$, and there exists a $\delta\geq 0$ s.t. the solution $\xx(\xx_i;t)$ to $\dot \xx=\fb_{q_i}$ satisfies \vspace*{-2mm}
      \begin{itemize}
         \item $\xx(\xx_i;t)\in D_{q_i}$ for all $t\in [0,\delta]$; and \vspace*{-2mm}
         \item $\xx(\xx_i;\delta)=\xx_{i+1}$\,.
      \end{itemize}
\end{itemize}

\subsection{Barrier certificates}
 Given an HS $\hh$ and a safety property $S$ (dually, an unsafe set $\UU$), the problem we considered is
 if $\RE_{\hh} \subseteq S$ (dually, $\RE_{\hh} \cap \UU=\emptyset$).
  Obviously, it is equivalent to $\forall q\in Q. \RH\!\! \upharpoonright_q \subseteq S_q$
  (dually, $\forall q\in Q. \RH\!\! \upharpoonright_q \cap \UU_q=\emptyset$),
  where $\RH\!\! \upharpoonright_q$ stands for all continuous states of
    $\RH$ projecting onto $q$.
     For this problem on CDSs, Prajna et al in \cite{PJ04,prajna2007}
      used the idea of Lyapunov functions for stability analysis in control theory
      to separate safe states from unsafe states by a barrier function with convexity,
       called \emph{\bc}.
      According to their definition, a barrier function $\varphi(\xx) \in \CALC^\omega[\RR^n]$ satisfies
        the following conditions: \vspace*{-4mm}
     \begin{enumerate}
     \item[i)] $\varphi(\xx)\le0$ for any point $\xx\in \Xi_q$; \vspace*{-2mm}
     \item[ii)] $\varphi(\xx)>0$ for any point $\xx\in \UU_q$; and \vspace*{-2mm}
     \item[iii)]
	$\forall \xx\in D_q. {\LIDER}_{{\fb}_q} \varphi(\xx)\le0$, where ${\LIDER}_{{\fb}_q} \varphi(\xx)=\frac{\partial \varphi}{ \partial \xx} {\fb}_q(\xx)$ is the Lie derivative of $\varphi$ with respect to the
	vector field ${\fb}_q$. \vspace*{-4mm}
\end{enumerate}

 Trivially to see, the existence of a barrier certificate is just a sufficient condition to guarantee the safety property to be verified.  Hence, using Prajna et al's approach, \oomit{all generated barrier certificates are polynomials.}
 one cannot claim the property does not hold if he/she fails to discover a polynomial barrier certificate.
 Actually, as observed in \cite{kong2013} by Kong et al, if condition iii)  is relaxed to the following
  iii'), one can synthesize barrier certificates with more expressiveness.
  Certainly, it is more likely to prove a safety property by using a more expressive barrier certificate, as it gives
  a tighter approximation of the reachable set. \vspace*{-4mm}
    \begin{enumerate}
     \item[iii')]
	${\LIDER}_{{\fb}_q} \varphi(\xx) - \gamma \varphi(\xx) \le0$, where $\gamma$ is a real number.
\end{enumerate}


\section{Revisiting Barrier Certificate Conditions}
\label{sec:bccondtion}
In this section, we investigate how to relax the condition of barrier certificates in a general way.

\subsection{Relaxed barrier certificate conditions for CDSs}
First of all,  we consider how to relax the condition i)-iii) of barrier certificates given in \cite{PJ04,prajna2007} for CDSs
in a general way.  To the end, we need to have a principle to justify when 
 a relaxed condition of barrier certificates is reasonable.
 An obvious principle is:
 \begin{description}
\item[Principle of Barrier Certificate (PBC):] Given a CDS $\mathcal{D}$  equipped with an initial set $\II$ and an unsafe set $\UU$, a barrier certificate should be a real-valued function $\varphi(\xx)$  such that
  $\varphi(\xx)\le0$ for any $\xx \in \mathcal{R}_{\mathcal{D}}$,  and
    $\varphi(\xx)>0$ for any point $\xx \in \UU$.
\end{description}
  Certainly, if there exists such a function
$\varphi(\xx)$, we can assert that
$\mathcal{R}_{\mathcal{D}} \cap \UU=\emptyset$, and $\phi(\xx)\leq 0$ is an invariant.
 However, such a principle cannot be effectively
checked in general,
so we have to strengthen the condition to make it effectively checkable, like in \cite{PJ04,prajna2007,kong2013}.
An interesting  problem is with which condition more expressive barrier certificates can be synthesized, but 
 the condition is still effectively checkable and
    satisfies \textbf{PBC}. We answer the problem by the following theorem.

\begin{theorem}[General Barrier Condition (\textbf{GBC})] \label{the:1}
  Given a CDS $\mathcal{D}$ equipped with
a domain $D$, an initial set $\II$ and an unsafe set $\UU$, if there is  a   function $\varphi(\xx)\in \CALC^\omega[\RR^n]$, a real
function $ \psi(\xx) \in \CALC^\omega[\RR]$
   such that
	\begin{align*}
		\forall \xx\in \CALX_0. & \varphi(\xx)\le 0,  \numberthis \label{cond:1} \\
		\forall \xx\in D. & \LIDER_f\varphi(\xx)-\psi(\varphi(\xx))\le0, \numberthis\label{cond:2}\\
		\forall \xx\in \UU. & \varphi(\xx)>0, \numberthis \label{cond:3} \\
		\xi>0 \Rightarrow & \theta(\xx(\xi))\le 0,  \mbox{ where } \theta(\xx(t)) \mbox{ is the  solution of} \\
      & \left\{ \begin{array}{l}
		           \theta(\xx(0))\le0, \\
		    \LIDER_{\fb} \theta(\xx)-\psi(\theta(\xx))=0,
               \end{array} \right.
       \numberthis  \label{cond:4}
	\end{align*}
	then $\mathcal{R}_\mathcal{D} \cap \UU=\emptyset$.
\end{theorem}

\begin{proof}
	Suppose $\xx_0\in \CALX_0$ and $\xx(t)$ is the corresponding solution of (\ref{eq:pro})
starting from $\xx_0$. Our goal is to  prove that for any function
	$\varphi(\xx(t))$ satisfying (\ref{cond:1})-(\ref{cond:4}), then
	\begin{equation}
		\forall \xi \ge0. \varphi(\xx(\xi))\le 0.
		\label{eq:sat}
	\end{equation}
	Let $g(\xx)=\LIDER_{\fb} \varphi(\xx)-\psi(\varphi(\xx))$, then by (\ref{cond:2})
	\begin{equation}
		\forall \xx\in \RR^n. g(\xx)\le0
		\label{eq:fung}
	\end{equation}
	Since $\frac{d\varphi(\xx(t))}{dt}=\frac{\partial \varphi}{\partial \xx} \frac{d\xx}{dt}= \frac{\partial \varphi}{\partial \xx} f(\xx)=\LIDER_f
	\varphi(\xx)
	$, we have
	\begin{equation}
		\left\{
			\begin{array}{l}
				\frac{d\varphi(\xx(t)) }{dt}-\psi(\varphi(\xx(t)))-g(\xx(t))=0\\
				\varphi(\xx(0))=\varphi(\xx_0)
			\end{array}\right.
			\label{eq:ode}
		\end{equation}
		Assume $\varphi(\xx(\xi))>0$, for some $\xi>0$.
		Let $\theta(\xx(t))$ be a function with
     \begin{equation}
			\left\{
				\begin{array}{l}
					\frac{d\theta(\xx(t)) }{dt}-\psi(\theta(\xx(t)))=0\\
					\theta(\xx(0))=\varphi(\xx_0)
				\end{array}\right.
				\label{eq:ode1}
			\end{equation}
			Let $\Theta=\{\xi\ \mid  \varphi(\xx(\xi))> \theta(\xx(\xi)),\xi\ge0 \}$.
			By (\ref{cond:4}), $\forall \xi>0. \theta(\xx(\xi))\le0$.  $\Theta$ is nonempty since the assumption. So there is a number $\mu$ s.t. $\mu=\inf(\Theta)$.
      Obviously, $\varphi(\xx(t))$, $\theta(\xx(t))$, $g(\xx(t)), \frac{d\varphi(\xx(t)) }{dt}$ and
			$\frac{d\theta(\xx(t)) }{dt}$ are analytic functions  w.r.t. $t$.  Thus
			$\varphi(\xx(\mu))=\theta(\xx(\mu))$. If
			$ g(\xx(\mu))<0$, then  $ \frac{d\varphi(\xx(t)) }{dt}|_{t=\mu}<\frac{d\theta(\xx(t)) }{dt}|_{t=\mu}$. Hence, $\exists \nu.\nu>\mu \wedge
			\forall \xi \in (\mu, \nu).$  $ \frac{d\varphi(\xx(t)) }{dt}|_{t=\xi}<\frac{d\theta(\xx(t)) }{dt}|_{t=\xi}$.
  Thus,
			$\forall \xi \in (\mu, \nu). \varphi(\xx(\xi))<\theta(\xx(\xi))$,  which  contradicts to the definition of $\mu$. So $g(\xx(\mu))=0$ and
			$\frac{d\varphi(\xx(t)) }{dt}|_{t=\mu}=\frac{d\theta(\xx(t)) }{dt}|_{t=\mu}$.  If there is a $k>1$ s.t.
			$\frac{d^k\varphi(\xx(t))}{dt^k}|_{t=\mu}<\frac{d^k\theta(\xx(t))}{dt^k}|_{t=\mu}$, and $\forall i<k,
			\frac{d^i\varphi(\xx(t))}{dt^i}|_{t=\mu}=\frac{d^i\theta(\xx(t))}{dt^i}|_{t=\mu} $,  then there is $\nu_1>\mu$ s.t.
			$\varphi(\xx(\xi))<\theta(\xx(\xi))$ for any $\xi \in (\mu,\nu_1)$, which contradicts to the definition of $\mu$. If
			$\forall k>1. \frac{d^k\varphi(\xx(t))}{dt^k}|_{t=\mu}=\frac{d^k\theta(\xx(t))}{dt^k}|_{t=\mu}$, then $\varphi(\xx(\xi))=\theta(\xx(\xi))$ for any $\xi \in
			\RR^+$, since $\varphi,\theta$ are analytic functions. So, the claim has been proved.
			Suppose for some $k>1$,
			$\frac{d^k\varphi(\xx(t))}{dt^k}|_{t=\mu}>\frac{d^k\theta(\xx(t))}{dt^k}|_{t=\mu}$ and $
			\forall i<k. \frac{d^i\varphi(\xx(t))}{dt^i}|_{t=\mu}=\frac{d^i\theta(\xx(t))}{dt^i}|_{t=\mu}$.  For all $i<k$, we simultaneously compute the $i$th derivatives of   the two sides of the first formulas of (\ref{eq:ode})  and
		 (\ref{eq:ode1}), and obtain
   $\frac{d^{i}\psi(\varphi(\xx(t)))
		}{dt^i}|_{t=\mu}=\frac{d^i\psi(\theta(\xx(t)))} {dt^i}|_{t=\mu},$ $ \frac{d^ig(\xx(t) ) }{dt^i }|_{t=\mu} =0$ for  $i<k-1$,  and \oomit{$ \frac{d^{k-1}\psi(\varphi(\xx(t)))
	}{dt^{k-1}}|_{t=\mu}=\frac{d^{k-1}\psi(\theta(\xx(t)))} {dt^{k-1}}|_{t=\mu},$}  $ \frac{d^{k-1}g(\xx(t))}{dt^{k-1}}|_{t=\mu}>0$. Thus, there is an $\delta>\mu$ s.t. $\forall \xi \in (\mu,\delta). g(\xx(\xi))>0$, which contradicts to the definition of $g(\xx)$. This completes the proof.
		\end{proof}

From now on, we call $\varphi$ in Theorem~\ref{the:1} a \bc \  of  $\mathcal{D}$. \vspace*{-4mm}

		\begin{remark}
			\label{ex:1}
\begin{itemize}
\item The application of Theorem~\ref{the:1} includes the following two steps: i) look for a function $\psi$ which satisfies condition (\ref{cond:4}); ii) similar to the work in \cite{kong2013},
     synthesize \bc \ according to the resulted conditions of
     (\ref{cond:1})-(\ref{cond:3}) by instantiating $\psi$  with the function obtained in the first step. \vspace*{-1mm} 
\item All barrier certificates that can be synthesized using the existing approaches  can also be synthesized according to
these conditions by instantiating $\psi$ to some specific functions satisfying condition (\ref{cond:4}).
  For instance, \textit{convex condition} in \cite{prajna2007} and \textit{differential invariant} in \cite{PlatzerClarke08}
  correspond to $\psi( \varphi)=0$, while
			\ec\ in \cite{kong2013} corresponds to  $\psi(\varphi)=\alpha\varphi$, where $\alpha \in \RR$. \vspace*{-4mm} 
\end{itemize}
		\end{remark}

The following lemma indicates that we can find a class of functions $\psi$ different from existing ones, satisfying condition (\ref{cond:4}). Thus, from which we can construct a class of relaxed conditions of barrier certificates by \textbf{GBC},
   that can be used to generate barrier certificates with different expressiveness.

		\begin{lemma} If
			 \begin{align*}
				\left\{ \begin{array}{l}
                 \frac{\partial \theta}{\partial t}-\alpha \theta-\beta\theta^2=0,   \\
				 \theta(0) \le0,
                   \end{array} \right. \numberthis \label{cond:10}
			    \end{align*}
  where $\alpha<0, \beta \in \RR$, then
			$\forall \xi>0. \, \theta(\xi)\le 0$.
			\label{lem:1}
		\end{lemma}
		\begin{proof}
			If $\beta \leq 0$, then from  (\ref{cond:10}) we have
			\begin{align*}
				\frac{\partial \theta}{\partial t} - \alpha \theta = \beta \theta^2 \leq  0 \numberthis \label{cond:11}
			\end{align*}
			So, the claim is guaranteed by Theorem 1 in \cite{kong2013}.

    Now, suppose  $\beta > 0$.
			Let $\lambda\in \RR$ with $\beta \lambda = \alpha$, and ${\theta}_0 = {\theta} (0)$, then
			\begin{align*}
				  & \frac{\partial \theta}{\partial t}=\alpha \theta +\beta \theta^2 \\
				\Rightarrow ~~ & \frac{\partial \theta} { \alpha \theta +\beta\theta^2}=\partial t \\
				\Rightarrow ~~ & \frac{d \theta}{\theta(\lambda+\theta)}=\beta dt \\
				\Rightarrow ~~ & \frac{1}{\lambda} (\frac{d\theta}{\theta}-\frac{d\theta}{ \lambda+\theta  })=\beta dt \\
            \end{align*}
             \begin{align*}
				\Rightarrow ~~ & ln \frac{\theta}{\lambda + \theta} = \lambda \beta t + c_0 = \alpha t + c_0 \\
				\Rightarrow ~~ & \frac{\theta}{\lambda + \theta} = {e} ^{\alpha t + c_0} \\
				\Rightarrow ~~ & \frac{\theta}{\lambda + \theta} = \frac{{\theta}_0}{\lambda + {\theta}_0} {e} ^{\alpha t} \\
				\Rightarrow ~~ & \theta = ( \frac{1}{1-\frac{{\theta}_0}{\lambda + {\theta}_0} e^{\alpha t}} -1) \lambda \numberthis \label{cond:12}
			\end{align*}
			As ${\theta}_0 \leq 0$, $\beta \lambda = \alpha$, $\beta >0$ and $\alpha <0$, we have
				$0 \leq \frac{{\theta}_0}{\lambda + {\theta}_0} < 1$ and  $e^{\alpha \xi} \leq 1$.
			So,
			\begin{align*}
				0 \leq \frac{{\theta}_0}{\lambda + {\theta}_0} e^{\alpha \xi} < 1,  \\
				\frac{1}{1-\frac{{\theta}_0}{\lambda + {\theta}_0} e^{\alpha \xi}} -1 \geq 0.
			\end{align*}
			By $ \beta \lambda = \alpha$, $\beta >0$ and $\alpha <0$,  it follows $\lambda <0$.
			From (\ref{cond:12}), we have
				$\forall \xi > 0. \theta ( \xi ) \leq 0 $.
		\end{proof}

\begin{remark}\label{remark:1}
			One can flexibly choose different relaxed conditions from the above class by setting
 different values to $\alpha$ and $\beta$ according to the following rules, that is illustrated in Fig. \ref{fig:1}:
  \begin{itemize}
   \item
   if  the value of $\alpha$  is smaller, then synthesized barrier certificates by the resulted
    condition from \textbf{GBC} are more expressive, and vice versa;
    \item if the value of $\beta$ is greater, then synthesized barrier certificates are more expressive, and
     vice versa.
   \end{itemize}
		\end{remark}

		\begin{figure}
			\centering
			\epsfig{file=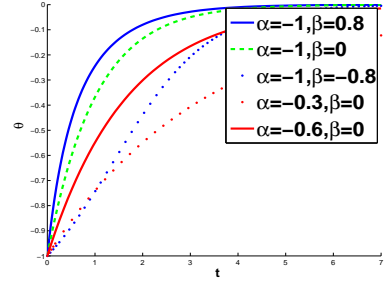,height=1.5in, width=2in}
			\caption{ Solutions of (\ref{cond:10}) with $\theta_0=-1$  on  different values of $\alpha,\beta$. }
			\label{fig:1}
		\end{figure}

The following example clearly indicates that one can synthesize some interesting barrier certificates with some relaxed
conditions from the above class, which cannot be discovered using the existing approaches.

\begin{example}  \label{ex:cont}
 Consider a CDS $\mathcal{D}_{\ref{ex:cont}}$ as follows: \\[2mm]
 \hspace*{1.5cm}  $\left\{ \begin{array}{l}
     \dot{x_1}=x_1^2-2x_1+x_2, \\
     \dot{x_2}=x_1+x_2^2-2x_2,
     \end{array} \right. $ \\[2mm]
  with $\II=\{(x_1,x_2) \mid 0.01-x_1^2-x_2^2\ge 0\},\UU=\{(x_1,x_2) \mid x_1^2+x_2^2-0.25\ge 0 \}$

        By Theorem \ref{the:1},  we can check that $\varphi=x_1^2+x_2^2-0.04$  is a barrier certificate
        w.r.t.
        $\psi(\theta)=-\theta+2\theta^2$ as follows: Let $g_0=0.01-x_1^2-x_2^2$, $g_1=x_1^2+x_2^2-0.25$.  Obviously, $-\varphi-g_0=0.03>0$, $\varphi-g_1=0.21>0$ and $-\LIDER_f(\varphi)-\varphi+2\varphi^2=2x^4
        - 2x^3 + 4x^2y^2 + 2.84x^2 - 4xy + 2y^4 - 2y^3 + 2.84y^2 + 0.0432$ is an \textbf{SOS},
         so the condition of Theorem \ref{the:1} is satisfied.

         On the other hand, we can show that
         there is no a barrier certificate $\varphi$ with $\\textit{deg}(\varphi)\le 2$ that can be synthesized by
         the condition given in \cite{kong2013}.
         Assume there is is a barrier certificate satisfying the condition of
             \cite{kong2013} of the form
              $$\varphi=a_{20}x_1^2+a_{11}x_1x_2+a_{02}x_2^2+a_{10}x_1+a_{01}x_2+a_{00}$$  w.r.t.
           $\psi(\theta)=\alpha\theta$, where $\alpha, a_{20}, a_{11}, a_{02}, a_{10}, a_{01}, a_{00} \in \RR$.
             Let $L=
            -\LIDER_f(\varphi)+\alpha\varphi$, so $L$ should be \textbf{SOS}.  From
            $\II$ and $\UU$, it follows that not all of  $a_{20},a_{11},a_{02}$ are equal to $0$. Suppose $a_{20}\neq 0$,
             then $L$ has a monomial $2a_{20}x_1^3$. Consider the value of $L$ over the set
              $\{ (\xi,0) \mid a_{20}\xi<0\}$, it will become negative when $|\xi|$ becomes large enough.
             Similarly, we can derive a contradiction in cases when $a_{11}\neq 0$ and $a_{02}\neq 0$.
             This means that our claim holds. \qed
 \end{example}

	\subsection{Combined barrier certificates} 
	Given a CDS $\mathcal{D}$ equipped with $D$, $\II$ and $\UU$, suppose $\varphi(\xx)$
  is a barrier certificate satisfying Theorem \ref{the:1} w.r.t. another function $\psi(\xx)$.
  Clearly, $\{\xx\ \mid\ \varphi(\xx)\le0\}$ is an over-approximation of $\RE_{\mathcal{D}}$,
   while $\{\xx\ \mid\ \varphi(\xx)> 0\}$ is an over-approximation of $\UU$.
   It is very common that in many cases we cannot find such a single barrier certificate to over-approximate
   the reachable set, but it can be achieved by combining several functions together.
   We call the combination of these functions a \emph{combined barrier certificate}.
   Actually, a similar problem on differential  invariants has been discussed in
	\cite{PlatzerClarke08,GT08,sankaranarayanan2010,LZZ11}.

Below, we discuss how to combine two
functions together to form a combined barrier certificate. For easing discussion, let's fix the aforementioned
CDS $\mathcal{D}$.

	\begin{lemma}
		\label{lem:init}
  $\{\xx\ \mid\ \chi(\xx)\le0\} $ is an over approximation of
		 $\RE_{\mathcal{D}}$, if
		\begin{align*}
			\forall \xx\in \CALX_0.\  & \chi(\xx)\le 0  \numberthis \label{cond:ib1} \\
			\forall \xx\in D.\ &  \LIDER_f \chi(\xx)-\psi(\chi(\xx))\le0 \numberthis\label{cond:ib2}\\
			\forall \xi. \, \xi>0 \Rightarrow& \theta(\xx(\xi))\le 0, \mbox{ where } \theta(\xx(t)) \mbox{ is the solution of }\\
           & \left\{ \begin{array}{l}
			\LIDER_f\theta(\xx)-\psi(\theta(\xx))=0, \\
             \theta(\xx(0))\le0,
              \end{array} \right.  \numberthis  \label{cond:ib4}
		\end{align*}
  where $\chi(\xx), \psi(\xx) \in \CALC^\omega[\RR^n]$.
	\end{lemma}
	\begin{proof}
		It can be proved similarly to Theorem \ref{the:1}.
	\end{proof}

	\begin{lemma} \label{lem:box} If there are    functions $\varphi(\xx),\chi(\xx)  \in \CALC^\omega[\RR^n]$
   with $\forall \xx\in \CALX_0.\chi(\xx)\le0$,
    $\psi(\xx) \in \CALC^\omega[\RR]$,  and a \textbf{SOS} polynomial
		$\delta$ \footnotemark \footnotetext{That is, $\delta$ can be represented by $f_1^2+...+f_n^2$, where
        $f_1,\ldots,f_n$ are polynomials.} such that
		\begin{align*}
			\forall \xx\in \CALX_0.\ & \varphi(\xx)\le 0  \numberthis \label{cond:b1} \\
			\forall \xx\in D.\ & \LIDER_f\varphi(\xx)-\psi(\varphi(\xx)) -\delta \chi(\xx) \le0 \numberthis\label{cond:b2}\\
			\forall  \xx\in \UU.\ & \varphi(\xx) >0 \numberthis \label{cond:b3} \\
			\forall \xi.\xi>0 \Rightarrow & \theta(\xx(\xi))\le 0, \mbox{ where } \theta(\xx(t)) \mbox{ is the  solution of }\\
          & \left\{ \begin{array}{l}
            \LIDER_f\theta(\xx)-\psi(\theta(\xx))=0, \\
			\theta(\xx(0))\le0,
             \end{array} \right. \numberthis  \label{cond:b4}
		\end{align*}
		then for every  trajectory $\tau$ of $\mathcal{D}$, we have
     $$(\forall \xi \ge0. \chi(\tau(\xi))\le0) \Rightarrow (\forall \xi\ge 0. \tau(\xi) \not \in \UU).$$
	\end{lemma}
	\begin{proof}
		We only need to prove $\forall \xi\ge 0. \varphi(\tau(\xi))\le0$.
		\begin{align*}
		  & \forall \xx\in \RR^n.\ \LIDER_f\varphi(\xx)-\psi(\varphi(\xx)) -\delta \chi \le0 \\
			\Rightarrow & \forall \xi  \ge 0.  \  \frac{\partial \varphi(\tau(t)) } {\partial t }|_{t=\xi}  -\psi(\varphi(\tau(\xi))) -\delta
			\chi(\tau(\xi)) \le 0 \\
           &  \hspace*{1.5cm} \mbox{ as } \LIDER_f\varphi(\xx)=\frac{\partial \varphi(\tau(t)) } {\partial t } \\
		\Rightarrow & \forall \xi  \ge0. \ 	\frac{\partial \varphi(\tau(t)) } {\partial t }|_{t=\xi}  -\psi(\varphi(\tau(\xi))) \le0  \\
    &  \hspace*{1.5cm}  \mbox{ as }   \forall \xi
		 \ge0. \
		 \chi(\tau(\xi))\le0.
	\end{align*}
	 Thus, by Theorem \ref{the:1}, the claim is trivially true.
	\end{proof}

	\begin{theorem}
		Let $\chi(\xx) \in \CALC^\omega[\RR^n]$ satisfy
		(\ref{cond:ib1})-(\ref{cond:ib4}).   If there are    functions $\varphi(\xx)\in \CALC^\omega[\RR^n],\psi(\xx)\in \CALC^\omega[\RR]$,
		\ and a \textbf{SOS} polynomial
		$\delta$, s.t. (\ref{cond:b1})-(\ref{cond:b4}) hold, then $\RE_{\mathcal{D}} \cap \UU =\emptyset$.
		\label{the:box}
	\end{theorem}

 \begin{proof} It is straightforward by Lemmas \ref{lem:init}\&\ref{lem:box}.
	\end{proof}

We will call the pair $(\chi,\phi)$ a combined barrier certificate.

Clearly, a single barrier certificate defined in Theorem \ref{the:1} can be seen as
a specific combined barrier certificate by letting $\chi=0$.
In addition,  actually, it is easy to prove that 
a combined barrier certificate forms a combined differential  invariant.
 \begin{corollary}
		$\chi\le 0 \wedge \varphi\le 0$ is a \textit{differential invariant} (the definition can be found in \cite{PlatzerClarke08}) of $\mathcal{D}$, which can guarantee its safety.
		\label{cor:1}
	\end{corollary}

\oomit{	\begin{corollary}
		Theorem \ref{the:1} is a special case of Theorem \ref{the:box} when set $\chi=0$.
		\label{cor:2}
	\end{corollary} }

	We use the following example to demonstrate the notion of combined barrier certificates
gives more  power to the verification of CDSs as well as HSs.

	\begin{example} Consider the following CDS $\mathcal{D}$
		\label{ex:4}
			\begin{equation*}
				\begin{bmatrix}
					\dot{x_1}           \\[0.3em]
					\dot{x_2}
				\end{bmatrix}	= \begin{bmatrix}
					2x_1-x_1x_2          \\[0.3em]
					2x_1^2-x_2\\[0.3em]
				\end{bmatrix}
			\end{equation*}
			with $\II=\{\xx\in \RR^2\ \mid \ x_1^2+(x_2+2) \le 1\}$ and $\UU=\{\xx\in \RR^2\ \mid x_2+(x_2-1)^2\le 0.09\}.$
	
 To prove its safety, by Theorem \ref{the:box}, we can synthesize a combined barrier certificate
   $(\chi,\varphi)$, see Fig. \ref{fig:6}, in which $\chi(\xx)=0$ is denoted by
    the red line and $\varphi(\xx)=0$ is denoted by the black line (their mathematical representations
    can be found in the appendix).
    In fact, we can prove  $\chi(\xx) \le0 \wedge \varphi(\xx)\le0 $ is indeed a differential invariant
    according to the definition given in \cite{PlatzerClarke08}, which can guarantee the unsafe set unreachable.

    Besides,  we can prove that neither of $\chi(\xx)$ nor $\varphi(\xx)$ is a barrier certificate
    in the sense of Theorem~\ref{the:1}. Furthermore, using the same values of  $\alpha,\beta$  and
    the degree bound as used in synthesizing the combined barrier certificate ($\chi(\xx), \varphi(\xx))$,
     we cannot obtain any single barrier certificate by Theorem~\ref{the:1}. \qed
	\begin{figure}
		\centering
		\epsfig{file=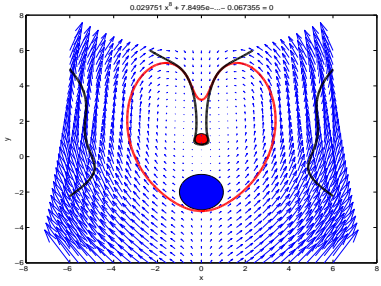,height=1.5in, width=2in}
		\caption{\small A combined barrier certificate for Example \ref{ex:4}}
		\label{fig:6}
	\end{figure}
\end{example}

	\subsection{Relaxed barrier certificate conditions  for HSs}
As discussed in \cite{kong2013},
the principle of the condition of 	
 barrier certificates $\Phi(\xx)$ for an HS $\hh = (Q, X, f, D, E, G, R, \Xi)$ w.r.t. a given
unsafe set $S^u$ should satisfy the following conditions:
\begin{itemize}
\item  $\Phi(\xx)$ consists of a set of functions $\{\varphi_q(\xx) \mid
	q\in Q\}$, each $\varphi_q(\xx)$ is a barrier certificate for CDS $\dot{\xx}=\fb_q$ equipped with
 the domain $D_q$, initial set $\Xi_q$ and unsafe set $S^u_q$;
 \item all the discrete transitions starting from every mode $q\in Q$ have to be
 taken into account in the
	barrier certificate condition
 so that $\Phi(\xx)$ can construct a global inductive invariant of $\hh$.
 \end{itemize}

Based on the discussions about barrier certificate conditions for CDSs as well as the above principle, we can accordingly
 revisit the condition of barrier certificates for HSs based on the following theorem:

	\begin{theorem} \label{the:2}
		Given an HS $\hh = (Q, X, f, D, E, G, R, \Xi)$ and an
unsafe set $S^u$, if there exists a set of non-negative real numbers
 $\{c_e \mid e\in E\}$, and a set of   functions
       $\{\varphi_q(\xx)\in \CALC^\omega[\RR^n] \mid q\in Q\}\cup  \{\psi_q(\xx)\in \CALC^\omega[\RR] \mid q\in Q\}$
       s.t.
		\begin{align*}
		&	\forall q\in Q \forall \xx\in \Xi_q.  \varphi_q(\xx)\le 0  \numberthis \label{cond:5} \\
		&	\forall q\in Q \forall \xx\in D_q.  \LIDER_{\fb_q}\varphi_q(\xx)-\psi_q(\varphi_q(\xx))\le0
                              \numberthis\label{cond:6}\\
		&	\forall q\in Q \forall \xx\in S^u_q. \varphi_q(\xx)>0 \numberthis \label{cond:7} \\
             \end{align*}
\begin{align*}
		&	\forall q\in Q \forall \xi. \xi>0 \Rightarrow  \theta_q(\xx(\xi))\le 0, \\
        & \hspace*{1.5cm} \mbox{ where } \theta_q(\xx(t)) \mbox{ is the  solution of } \\
   & \hspace*{.5cm}  \left\{ \begin{array}{l}
             \LIDER_{\fb_q} \theta_q(\xx)-\psi_q(\theta_q(\xx))=0, \\
             \theta_q(\xx(0))\le 0,
             \end{array} \right. \numberthis  \label{cond:8} \\
   & \forall e\in E \forall \xx \in G(e) \forall \xx'\in R(e)(\xx). \\
   &    \hspace*{1.5cm}  c_{e}\varphi_{S(e)}(\xx)-\varphi_{T(e)}(\xx')\ge 0,  \numberthis \label{cond:9}
		\end{align*}
     then $\RH \cap S^u =\emptyset$, where $S(e)$ and $T(e)$ respectively are
      the source and target modes of jump $e$.
	\end{theorem}

	Similarly, based on Theorem~\ref{the:box} and Theorem~\ref{the:2}, we can revisit the condition of combined
   barrier certificates for HSs as follows:
	\begin{theorem} \label{the:4}
    	Given an HS $\hh = (Q, X, f, D, E, G, R, \Xi)$ and an
unsafe set $S^u$, if there exists a set of non-negative real numbers
 $\{c_{e,1},c_{e,2},c_{e,3},c_{e,4} \mid e\in E\}$, a set of \textbf{SOS} polynomials $\{\delta_q \mid q\in Q\}$,
    and a set of functions
       $\{\varphi_q(\xx), \chi_q(\xx) \in \CALC^\omega[\RR^n] \mid q\in Q\}
        \cup \{ \psi_{q,1}(\xx), \psi_{q,2}(\xx)  \in \CALC^\omega[\RR] \mid q\in Q\}$
       s.t.
		\begin{align*}
		&	\forall q\in Q \forall \xx\in \Xi_q. \chi_q(\xx) \le 0  \numberthis \label{cond:bn10} \\
		&	\forall q\in Q \forall \xx\in D_q. \LIDER_{\fb_q}\chi_q(\xx)-\psi_{q,1}(\chi_q(\xx))\le0
               \numberthis\label{cond:bn11}\\
		 & \forall q\in Q \forall \xx\in \Xi_q. \varphi_q(\xx) \le 0  \numberthis \label{cond:b10} \\
		&	\forall q\in Q \forall \xx\in D_q. \LIDER_{\fb_q}\varphi_q(\xx)-\psi_{q,2}(\varphi_q(\xx)) -\delta_q\chi_q \le0 \numberthis\label{cond:b11}\\
			& \forall q\in Q \forall \xx\in S^u_q. \varphi_q(\xx)>0 \numberthis \label{cond:b12} \\
		 & \forall q\in Q \forall \xi. \xi>0 \Rightarrow \theta_q(\xx(\xi))\le 0,  \\
        & \hspace*{1.5cm} \mbox{ where } \theta_q(\xx(t)) \mbox{ is the  solution of} \\
        & \hspace*{.5cm} \left\{ \begin{array}{l}
                \LIDER_{\fb_q} \theta_q(\xx)-\psi_{q,1}(\theta_q(\xx))=0, \\
                \theta_q(\xx(0))\le0,
                 \end{array} \right.  \numberthis  \label{cond:b13} \\
         & \forall q\in Q \forall \xi. \xi>0 \Rightarrow \theta_q'(\xx(\xi))\le 0, \\
        & \hspace*{1.5cm} \mbox{ where } \theta_q'(\xx(t)) \mbox{ is the  solution of} \\
        & \hspace*{.5cm} \left\{ \begin{array}{l}
                \LIDER_{\fb_q} \theta_q(\xx)-\psi_{q,2}(\theta_q'(\xx))=0, \\
                \theta_q'(\xx(0))\le0,
                 \end{array} \right.  \numberthis  \label{cond:bn13} \\
		& \forall e\in E \forall \xx \in G(e)\forall \xx'\in R(e)(\xx). \\
       & \hspace*{1.5cm} c_{e,1}\varphi_{S(e)}(\xx)-\varphi_{T(e)}(\xx')\ge0, \numberthis \label{cond:bn14} \\
       & \forall e\in E \forall \xx \in G(e)\forall \xx'\in R(e)(\xx). \\
       & \hspace*{1.5cm} c_{e,2}\varphi_{S(e)}(\xx)-\chi_{T(e)}(\xx')\ge0, \numberthis \label{cond:b14} \\
  	  & \forall e\in E \forall \xx \in G(e)\forall \xx'\in R(e)(\xx). \\
       & \hspace*{1.5cm} c_{e,3}\chi_{S(e)}(\xx)-\varphi_{T(e)}(\xx')\ge0, \numberthis \label{cond:bn15} \\
       & \forall e\in E \forall \xx \in G(e)\forall \xx'\in R(e)(\xx). \\
       & \hspace*{1.5cm} c_{e,4}\chi_{S(e)}(\xx)-\chi_{T(e)}(\xx')\ge0, \numberthis \label{cond:b15}
		\end{align*}
	 then $\RH \cap S^u =\emptyset$, where $S(e)$ and $T(e)$ are respectively
      the source and target modes of the jump $e$.	
	\end{theorem}

\section{Discovering Relaxed Barrier Certificates by \sdp}\label{sec:alg}
Theorems~\ref{the:1}\&\ref{the:box} (respt. Theorems \ref{the:2}\&\ref{the:4}) provide relaxed conditions which 
  can guarantee a function (a combination of two functions) to be a (combined) barrier certificate for a CDS (resp. an HS),  but these theorems  do not provide any constructive method to synthesizing (combined) barrier certificates.
In this section, we discuss how to exploit {\sdp} techniques \cite{Parrilo00,Parrilo03}  to construct (combined) barrier certificates
 from these relaxed conditions, which is inspired by previous work e.g. \cite{JM98,PJ04,prajna2007,
sloth2012,kong2013}.

Thus, we briefly review  {\sdp} first.

	\subsection{{\sdp}}
	We use  $\textit{Sym}_n$ to denote the set of  $n\times n$ real symmetric matrices, and $\textit{deg}(f)$  the highest total degree  of $f$ for a given polynomial
	$f$.

	\begin{definition}[Positive semidefinite matrices] A matrix $M\in \textit{Sym}_n$ is called \emph{positive semidefinite}, denoted by  $M\succeq 0$, if
		$\xx^TM\xx\ge 0$  for all $\xx\in \RR^n$.
	\end{definition}

	\begin{definition}[Inner product]\label{def:inner}
		The {\em inner product} of two matrices $A=(a_{ij}),B=(b_{ij})\in \RR^{n\times n}$, denoted by
		$\left<A,B\right>$, is defined by
		$\textit{Tr}(A^TB)=\sum _{i,j=1}^na_{ij}b_{ij}$.
	\end{definition}

	\begin{definition}[Semidefinite programming ({\sdp})]\label{def:sdp}
		The standard (primal) and dual forms of a {\sdp} are respectively given in the following:
		\begin{eqnarray}
			p^* &=& \inf_{X\in \textit{Sym}_n}\left<C,X\right> \mbox{ s.t. } X\succeq 0,\ \left<A_j,X\right>=b_j\ \label{eq:primal} \\
                & & ~~~~~~~~~~~~~~~~~~~~~~ (j=1,\ldots,m)  \nonumber \\
			d^* &= & \sup_{ y\in \RR^m} \mathbf{b}^T\yy \  \mbox{ s.t.} \ \sum_{j=1}^m y_jA_j +S=C,\  S \succeq 0,
			\label{eq:dual}
		\end{eqnarray}
		where  $C,A_1,\ldots,A_m,S\in \textit{Sym}_n$ and $\mathbf{b}\in \RR^m$.
	\end{definition}

	There are many  efficient algorithms to solve
	{\sdp} such as interior-point method. We present a basic path-following algorithm to solve (\ref{eq:primal}) in Algorithm~\ref{alg:inter}.

	\begin{definition}[Interior point for \sdp]
		\begin{eqnarray*}
			\textit{intF}_p & = & \left\{ X: \left<A_i,X\right>=b_i\ (i=1,\ldots,m),\ X\succ 0 \right\}, \\
			\textit{intF}_d & = & \left\{ (\mathbf{y},S): S=C-\sum_{i=1}^mA_iy_i\succ 0 \right\}, \\
			\textit{intF} & = & \textit{intF}_p\times \textit{intF}_d.
		\end{eqnarray*}
	\end{definition}

	Obviously, $\left<C,X \right> -\mathbf{b}^T\yy=\left<X,S\right> \geq 0 $ for all $(X,\yy,S)\in \textit{intF} $. Especially, we have $ d^*\le p^*$. So the soul of interior-point method to compute $p^*$
	is to reduce $\left<X,S\right>$ incessantly and meanwhile guarantee $(X,\yy,S)\in \textit{intF}$.
	\begin{algorithm}[!htb] \label{alg:inter}
		\SetKwData{Left}{left}\SetKwData{This}{this}\SetKwData{Up}{up}
		\SetKwFunction{Union}{Union}\SetKwFunction{FindCompress}{FindCompress}
		\SetKwInOut{Input}{input}\SetKwInOut{Output}{output}
		\Input{ $C$, $A_j, b_j\ (j=1,\dots,m)$ as in (\ref{eq:primal}) and a threshold $c$ }
		\Output{ $p^* $ }
		\SetAlgoLined
		\BlankLine
		Given a $(X,\yy,S)\in \textit{intF}$ with $XS=\mu I$\;
		\tcc{ $\mu$ is a positive constant and $I$ is the identity matrix. }
		\While{$\mu>c $  }{
			$\mu=\gamma\mu$\;
			\tcc{$\gamma$ is a fixed positive constant less than one}
			use  {\bf Newton iteration} to solve $(X,\yy,S) \in \textit{intF}$ with $XS=\mu I$\;
		}
		\caption{ {\tt Interior\_Point\_Method}
	}
	\end{algorithm}

\subsection{Symbolic checking} \label{sec:pro}

Please be noted that because of the error caused by numeric computation in {\sdp}, in particular,
  a threshold $c$ upon which {\sdp} depends, it may happen that the (combined) barrier certificates computed by {\sdp} are not
  real ones, or some real (combined) certificates satisfying the condition cannot be computed or are determined as
  false ones.
   For example, considering Example \ref{ex:4}, if we encode the condition derived from Theorem \ref{the:1} as a {\sdp}, then call SOSTOOLS\footnote{SOSTOOLS is of version v2.04 with MATLAB R2011b.}  \cite{prajna2005}, and obtain the output is: \\[2mm]
 \hspace*{1cm}  \textsf{ `` feasratio: 1.0000;  pinf: 0;  dinf: 0;
   numerr: 0''.} \\[2mm]
    This indicates that the tool does discover a barrier certificate.
   However, after showing the result in Fig. \ref{fig:8}, it is easy to find that the black line in Fig. \ref{fig:8} does not  satisfy condition (\ref{cond:2}), as
    some vectors cross it into the area which contains unsafe set.

   So, we have to take the numerical error into account when using these {\sdp} tools.  Our experience is: \vspace*{-4mm}
   \begin{itemize}
   \item   larger the size of matrix $X$  is, larger the error due to {\sdp}, so it is more likely to
      obtain a false (combined) barrier certificate; \vspace*{-1mm}
   \item higher the degree of undetermined
   polynomials as predefined templates of barrier certificates is,
   larger the error due to {\sdp}; \vspace*{-1mm}
   \item one can synthesize combined barrier certificates with lower degrees  by Theorem~\ref{the:box}
      than by Theorem \ref{the:1}. \vspace*{-2mm}
 \end{itemize}

	\begin{figure}
		\centering

		\begin{overpic}[width=2in,height=1.5in]{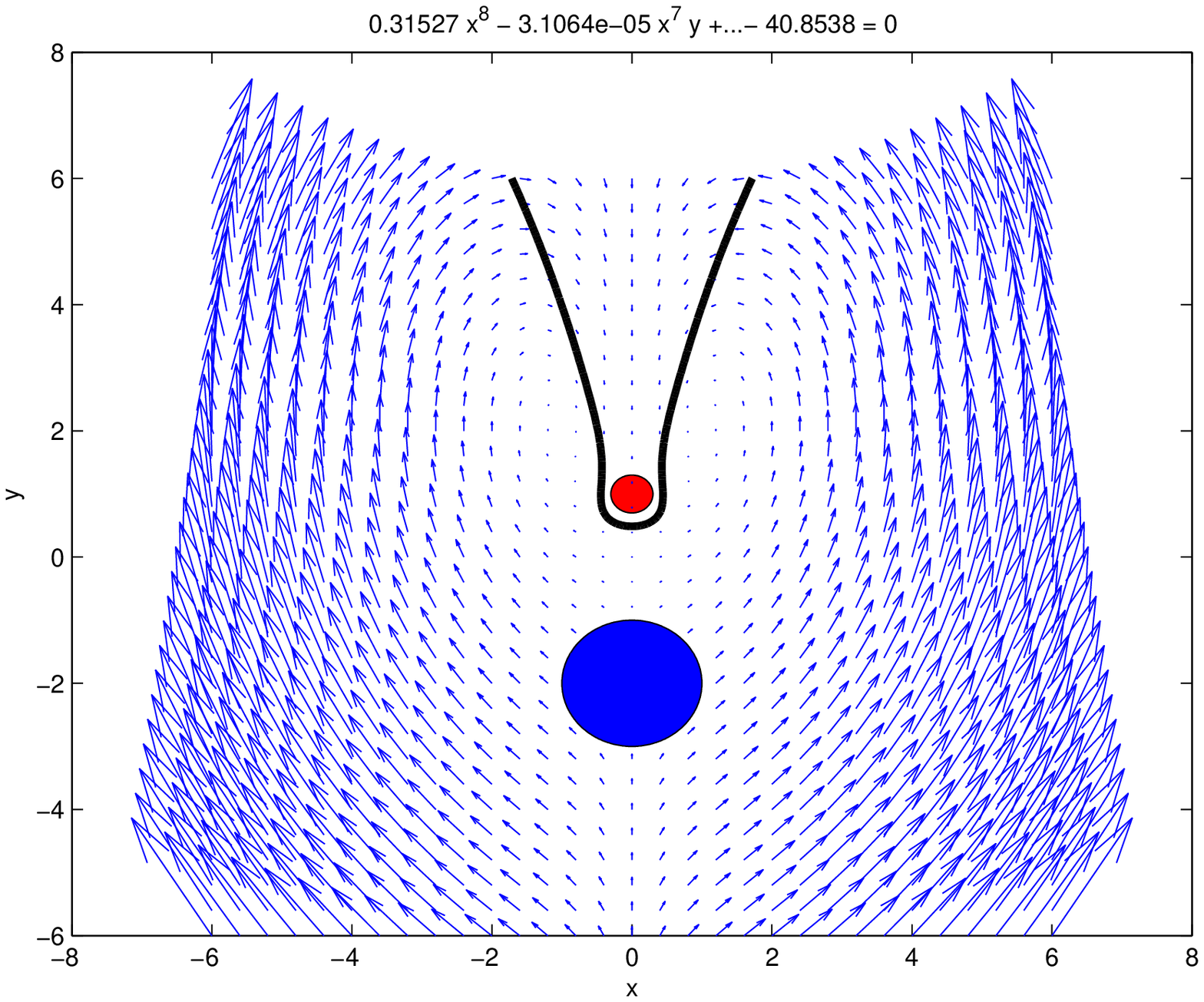}
			\put(65,43){\includegraphics[width=20mm,height=20mm]{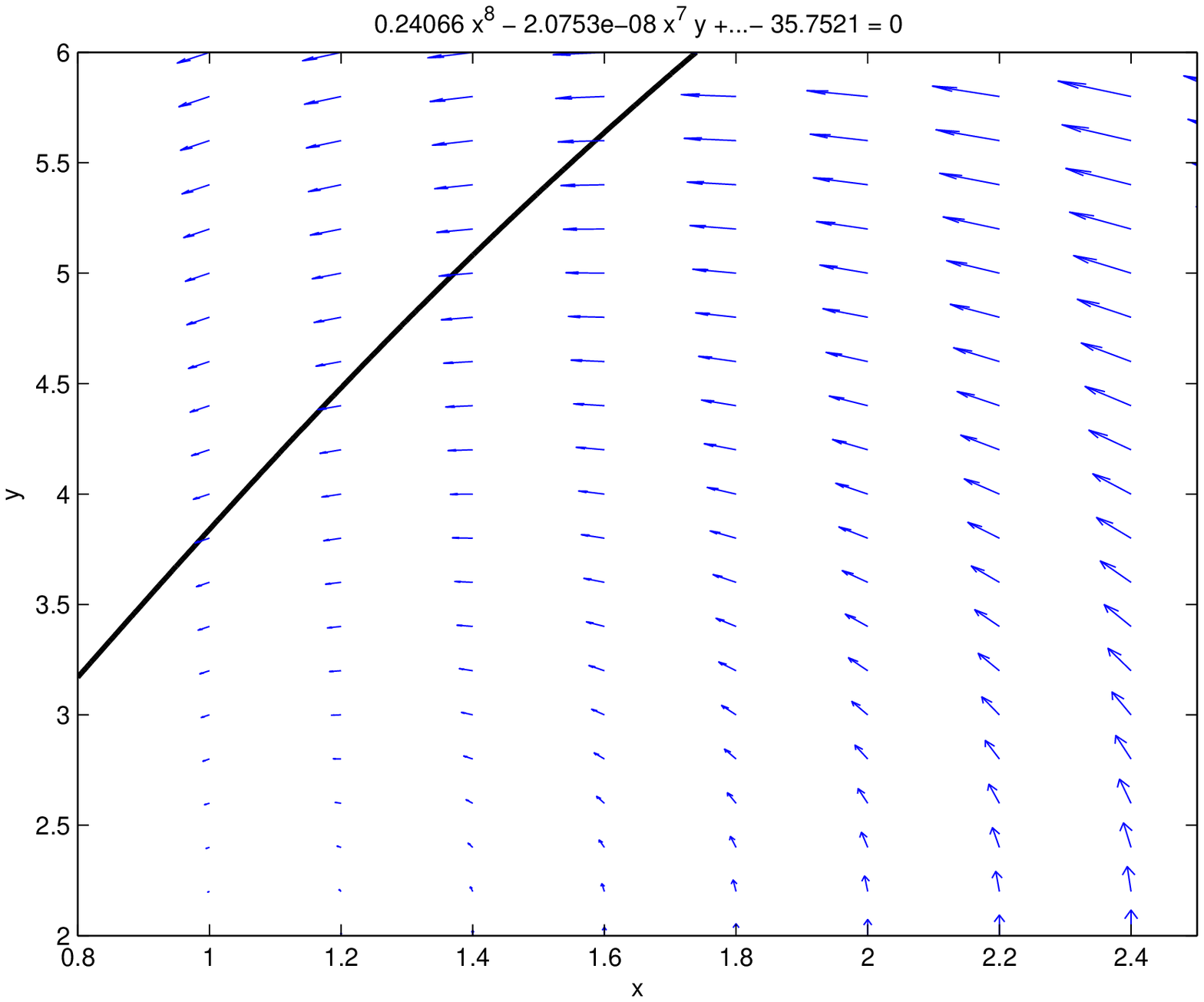}}
		\end{overpic}
		\caption{A false barrier certificate of Example~\ref{ex:4} due to numeric errors}
		\label{fig:8}
	\end{figure}

It is  absolutely necessary to guarantee the soundness of the approaches to
  the verification of HSs. But the approach based on {\sdp} to synthesize (combined) barrier certificates
  according to these relaxed conditions, may be unsound because of the error caused by
  numeric computation. Below, we advocate to apply symbolic computation techniques to check
   if the synthesized (combined) barrier certificates are real ones, which is hinted by our previous work \cite{DXZ2013}.
  \begin{problem}
   For $f\in \RR[\xx]$,  if $\forall \xx\in \RR^n.f(\xx)\ge0$ ?
  \label{pro:1}
 \end{problem}
 Checking the constraints in Theorems \ref{the:1}\&\ref{the:box}\&\ref{the:2}\&\ref{the:4} are obviously instances of
 Problem \ref{pro:1}.  A lot of work has been done on Problem \ref{pro:1}. We choose an
 exact method based on an improved Cylindrical Algebraic Decomposition(CAD) algorithm \cite{cpsd} for the checking, and
 call the  tool  \psd\  in the experiments, which
 implements the algorithms in \cite{cpsd}. The \psd\ returns {\bf True} when the input polynomial is positive semidefinite and {\bf False} otherwise.

 \begin{remark}
 One may doubt the efficiency of the above symbolic checking since the complexity of CAD is $O(2^{2^n})$ in general, where $n$ is the number of variables. However, please note that
 Problem \ref{pro:1} is a special case of quantifier elimination. One of the main contributions of \cite{cpsd} is an improved algorithm for solving
 Problem \ref{pro:1}. Although the improved algorithm cannot be proved with a lower complexity theoretically, it has been shown that it does avoid
 many heavy resultant computation. So, in practice, especially in the case where the number of variables are greater than 2, \psd \  is much faster
 than any general CAD tool. Please see \cite{cpsd} for details. In our experience, \psd\ can finish checking in few seconds when $\textit{deg}(f)$ is no larger than 6 and the number of variables in $f$ is less than $5$, which is enough for many problems.
 \end{remark}

\subsection{Algorithms}
We can sketch the basic steps of the algorithm to construct (combined) barrier certificates using {\sdp}  as follows: \vspace*{-4mm}
\begin{description}
\item[Step 1:] predefine parametric polynomial templates with a degree bound as possible candidates of (combined) barrier certificates; \vspace*{-1mm}
\item[Step 2:] derive constraints on the parameters of these parametric polynomial templates from the considered relaxed barrier condition;\vspace*{-1mm}
\item[Step 3:] reduce all the constraints on the parameters to a {\sdp};\vspace*{-1mm}
\item[Step 4] apply some {\sdp} solver to solve the resulted {\sdp} and obtain instantiations of
      these parameters. \vspace*{-4mm}
\end{description}
In the above procedure, for most of the constraints on parameters, we only need to
consider  how to reduce $p \ge 0$  ($p \le0$) to $p= \delta$ ($-p= \delta$),
 where $p$ is a polynomial and $\delta$ is a undetermined  \textbf{SOS} polynomial. In the literature,
  there is lot of work on this, please refer to  \cite{JM98,PJ04,prajna2007,
sloth2012,kong2013,DXZ2013} for the detail.

The hardest part is how to reduce the constraints that contain $\psi$, $\chi$, $\psi_q$,
  $\psi_{q,i}$, or $\chi_q$, as they may contain the product of
two or more parametric polynomials after replacement, which result in non-linear expressions on parameters, that cannot be seen
as a {\sdp} any more. For instance, let $\psi=\theta+\theta^2$, and $\theta = ax_1+bx_2$ be a template of barrier certificates.
 By Theorem~\ref{the:1}, the constraint derived from condition (\ref{cond:2}) will contain expression
  $(ax_1+b x_2)+(ax_1+bx_2)^2$, which cannot be reduced to a {\sdp} directly.

  To address this issue, we explore the iterative approach proposed in \cite{PJ04} which can handle
  a constraint containing the product of two parametric polynomials.
  We demonstrate the basic idea of the iterative approach by presenting Algorithm~\ref{alg:iterater} based on
  which for the following problem.

	\begin{problem} \label{pro:2}
		Suppose $\CALX_0,\UU,\fb,\psi$ are given, where $\psi$  satisfies (\ref{cond:4}), our goal is
	  to find a $\varphi$ which satisfies (\ref{cond:1})- (\ref{cond:3}).
	\end{problem}

	\begin{algorithm}[!htb] \label{alg:iterater}
		\SetKwData{Left}{left}\SetKwData{This}{this}\SetKwData{Up}{up}
		\SetKwFunction{Union}{Union}\SetKwFunction{FindCompress}{FindCompress}
		\SetKwInOut{Input}{input}\SetKwInOut{Output}{output}
		\Input{ $\CALX_0, \UU, \fb, \psi(\theta)=\sum_{i=0}^sa_i\theta^i$, where $\psi(\theta)$ satisfies (\ref{cond:4})}
		\Output{ $\theta'$ which satisfies (\ref{cond:1})-(\ref{cond:3}) }
		\SetAlgoLined
		\BlankLine
		$\theta'=0$\;
		$j=0$\;
		\While {$ j\le s$ }{
			$\psi'=\sum_{i=0}^ja_i\theta \theta'^{i-1} $\;
			Use a {\sdp} tool to solve the resulted Problem~\ref{pro:2} by replacing
                 $\psi$ with $\psi'$\;
			Denote the result of the above step by $\theta'$\;
			$j=j+1$\;
		}

		\caption{ {\tt Iterative Algorithm for Problem \ref{pro:2}}
	}
	\end{algorithm}

\section{Experimental Results} \label{sec:exp}
In this section, we demonstrate our approach by some examples.

		\begin{example}[modify example of \cite{kong2013new}] \label{ex:31}
			Consider a CDS $\mathcal{D}_{\ref{ex:31}}$ as follows:
			\begin{equation*}
				\begin{bmatrix}
					\dot{x_1}           \\[0.3em]
					\dot{x_2}
				\end{bmatrix}	= \begin{bmatrix}
					2x_1-x_1x_2          \\[0.3em]
					2x_1^2-x_2\\[0.3em]
				\end{bmatrix}
				\label{eq:kong13}
			\end{equation*}
			with  $\Xi_0=\{\xx\in \RR^2\ \mid \ x_1^2+(x_2+2) \le 1\}$ and $\UU=\{\xx\in \RR^2\ \mid x_2+(x_2-5.2)^2\le 0.81\}$.

		 No polynomial barrier certificates can be synthesized using the existing approaches for the
verification of $\mathcal{D}_{\ref{ex:31}}$, except for the one in \cite{kong2013new} with which a polynomial barrier certificate $\varphi(\xx)$ of degree $8$  was discovered.
        By setting $\alpha=-4$ and $\beta=1.5$, by the corresponding relaxed condition by \textbf{GBC},
        it is easy to synthesize a polynomial barrier
		certificate of degree $6$, see Fig. \ref{fig:3} (also see the appendix).
      In contrast, by setting $\beta=0$, the corresponding resulted relaxed condition is
       degenerated to the case considered in \cite{kong2013new}. But unfortunately,
       we can not synthesize an appropriate barrier certificate from the conditions,
       see Fig. \ref{fig:2}. \qed
       
       \medskip  \medskip 

		\begin{figure}[ht]
                 \begin{minipage}[b]{0.45\linewidth}
				\centering

				\begin{overpic}[width=38mm,unit=1mm]{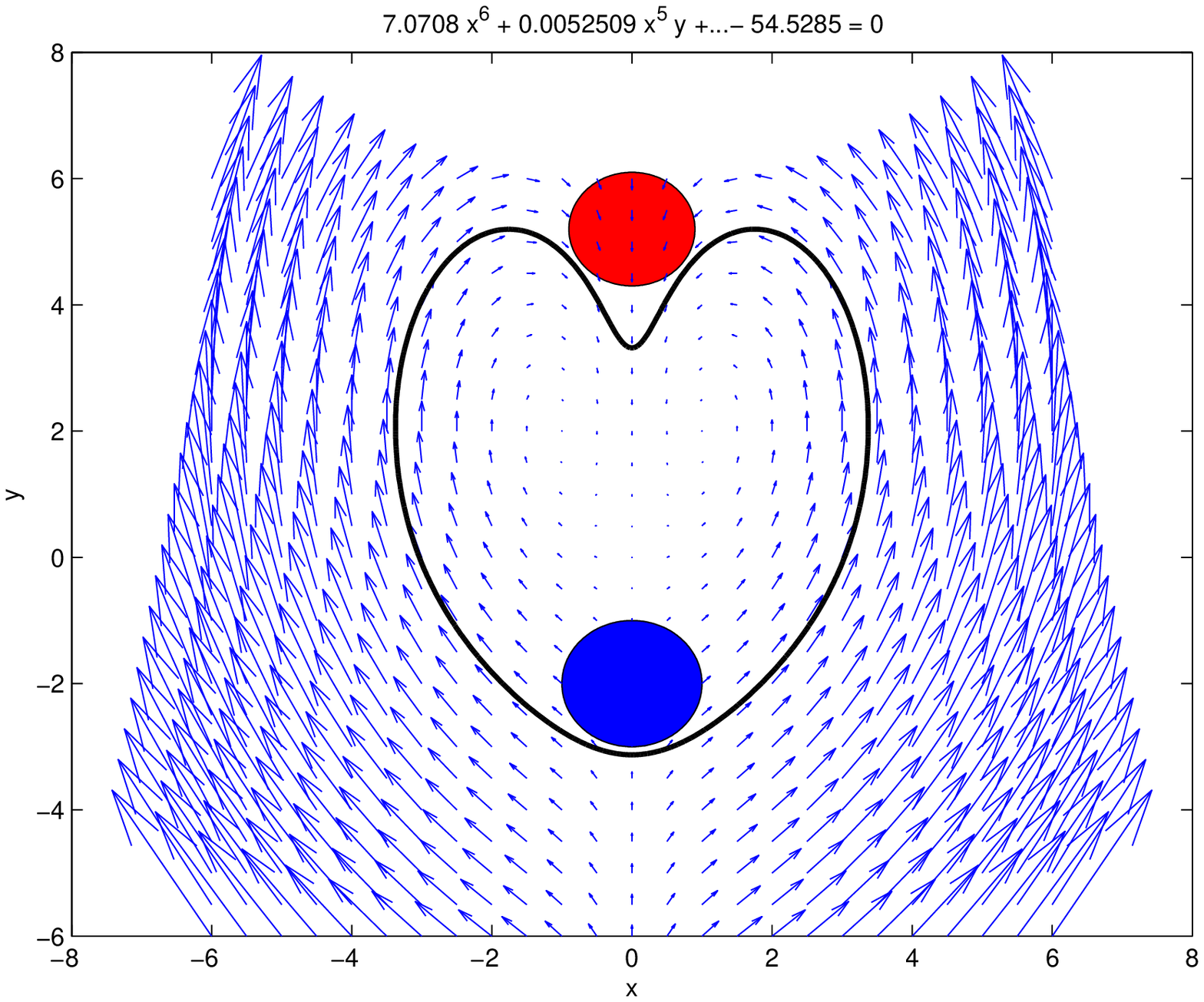}
					\put(60,50){\includegraphics[width=20mm]{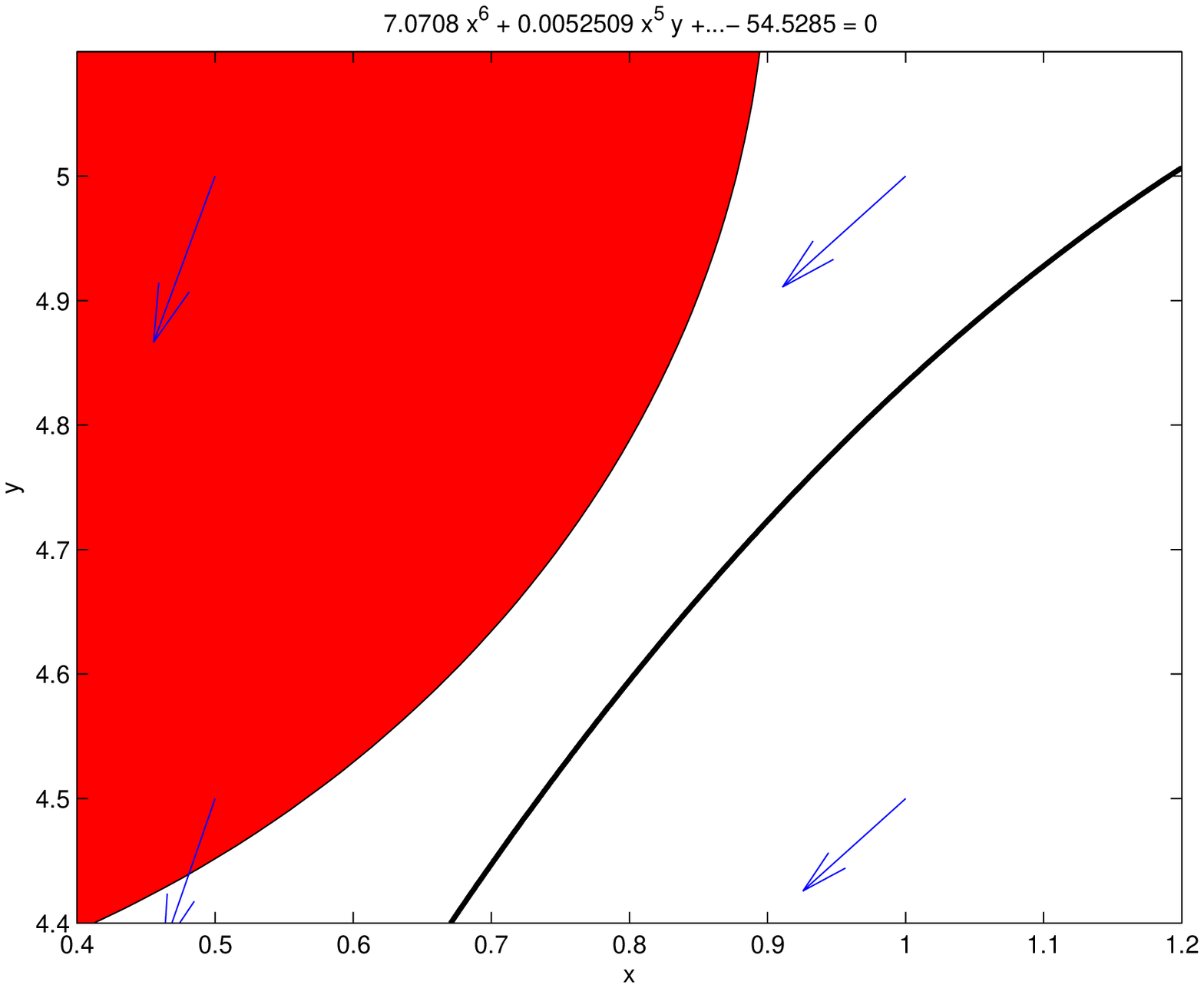}}
				\end{overpic}
				\caption{\small $\alpha=-4,\beta=1.5$}
				\label{fig:3}
			\end{minipage}
                \hspace{0.5cm}
			\begin{minipage}[b]{0.45\linewidth}
				\centering
				\begin{overpic}[width=38mm,unit=1mm]{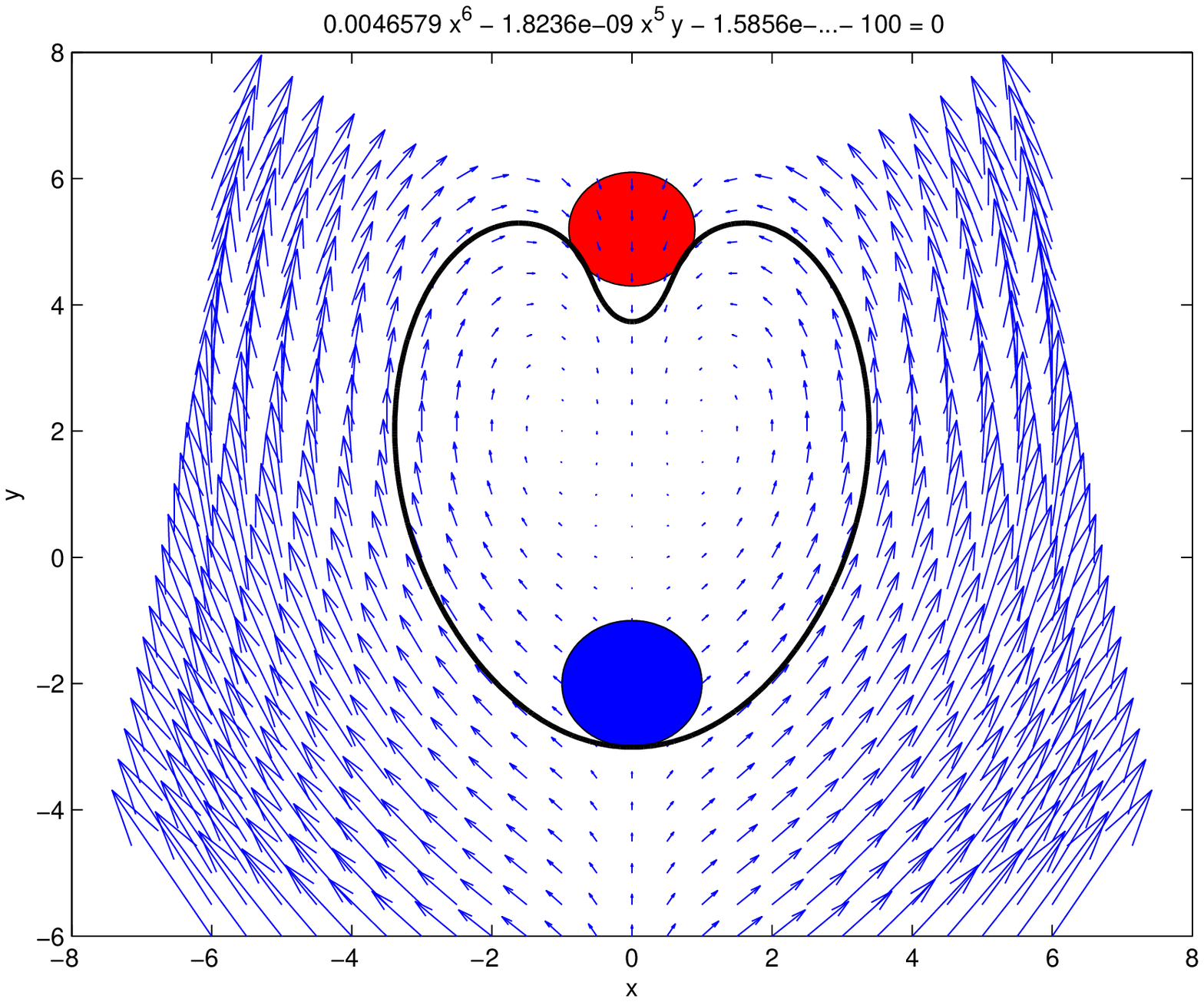}
					\put(60,50){\includegraphics[width=20mm]{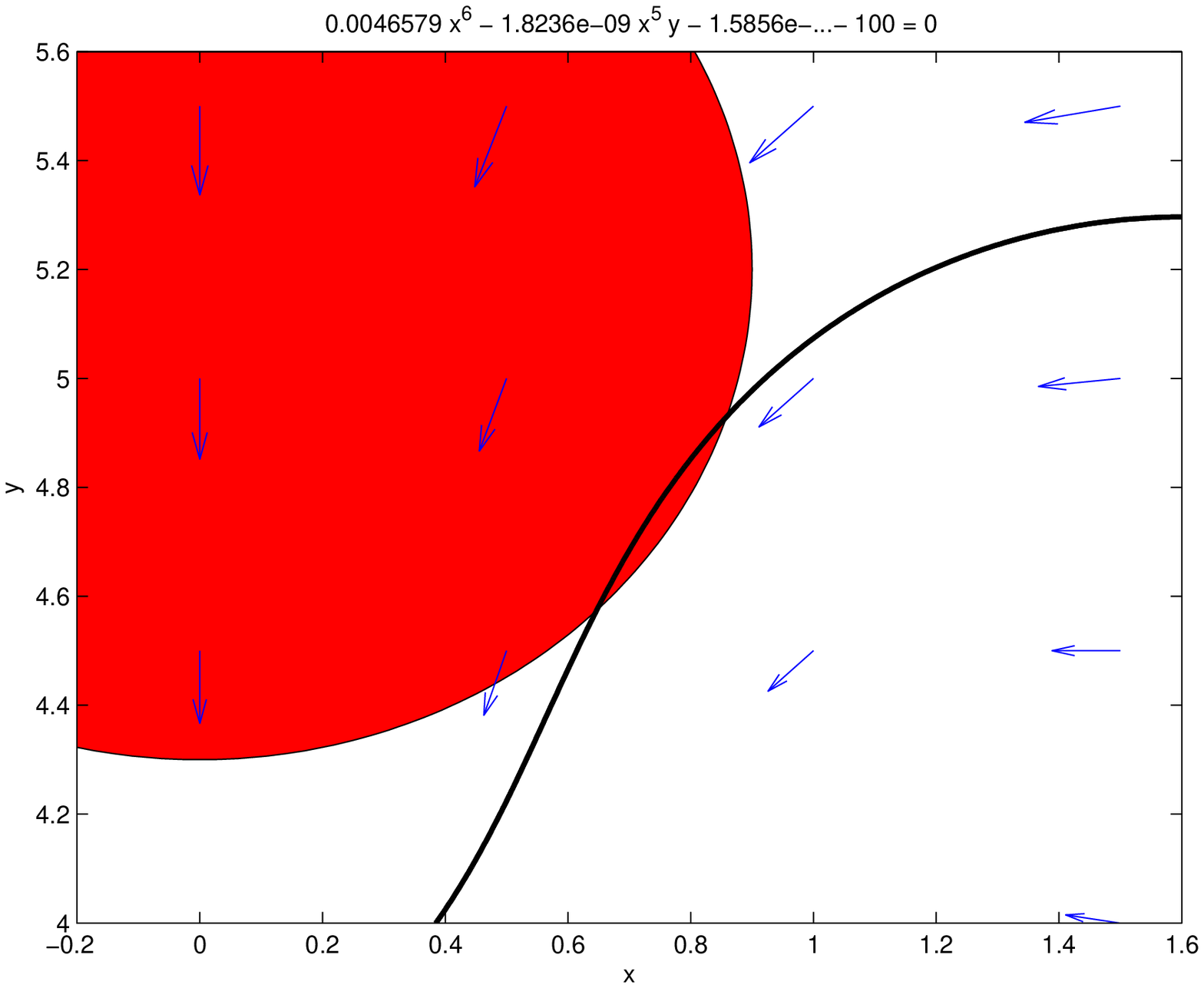}}
				\end{overpic}
				\caption{\small $\alpha=-4,\beta=0$}
				\label{fig:2}
			\end{minipage}	
		\end{figure}
	\end{example}

	\begin{example}
		\label{ex:5}
   Consider the following CDS $\mathcal{D}_{\ref{ex:5}}$
    	\begin{equation*}
		\begin{cases}
			\dot{x_1} =x_2  \\
			\dot{x_2} = 2x_1 - x_2 - x_1^2 x_2 - x_1^3
		\end{cases}
	\end{equation*}
		with $\II= \{ x \in \RR^2 \mid ( x_1 + 1 )^2 + (x_2-2)^2 \leq 0.16 \}$  and
		$\UU = \{ x \in \RR^2 \mid ( x_1-1 )^2 + x_2 ^2 \leq 0.04 \}$.
  Let  $g_0=0.16- ( x_1 + 1 )^2 - (x_2-2)^2,g_1=0.04- ( x_1 - 1 )^2 -  x_2 ^2$.
  In order to prove $\RE_{\mathcal{D}_{\ref{ex:5}}} \cap \UU =\emptyset$,
    according to Theorem~\ref{the:1}, using the above procedure,
     we can obtain the following polynomials:

     {\scriptsize $
       \begin{array}{ll}
    \varphi  = & -0.91253x_1^2 + 0.40176x_1x_2 + 1.3603x_1 + 0.13922x_2^2  \\
	           &- 1.0308x_2 - 0.27657,\\
	\chi =   & 0.19394x_1^4 + 0.29363x_1^3x_2 - 0.1696x_1^3 + 0.091674x_1^2x_2^2 \\
               &- 0.2317x_1^2x_2 - 1.3805x_1^2 + 0.056453x_1x_2^3 - 0.14904x_1x_2^2 \\
			   & + 0.096278x_1x_2 + 1.7932x_1 + 0.070488x_2^4 - 0.063002x_2^3 \\
			   & + 0.48804x_2^2 - 1.1726x_2 - 0.38201 \\

	\delta  = & 0.1956x_1^4 + 0.23674x_1^3x_2 - 0.13109x_1^3 + 0.14603x_1^2x_2^2 \\
	          &- 0.16935x_1^2x_2 + 1.0686x_1^2 + 0.35005x_1x_2^3 - 0.29307x_1x_2^2 \\
			  & - 0.5897x_1x_2 - 1.8943x_1 + 0.26073x_2^4 - 0.23047x_2^3 \\
			  & + 0.027813x_2^2 + 0.64131x_2 + 1.7118, \\
	 u_1  =  & 0.47292x_1^4 + 0.03761x_1^3x_2 - 0.15676x_1^3 + 0.45935x_1^2x_2^2\\
	         & + 0.13126x_1^2x_2 + 0.26007x_1^2 + 0.0766x_1x_2^3 - 0.02395x_1x_2^2\\
			 & + 0.045239x_1x_2 + 0.068505x_1 + 0.33983x_2^4 + 0.17729x_2^3 \\
			 & + 0.4338x_2^2 + 0.054172x_2 + 0.37428 \\
	u_2 =     & 0.45008x_1^4 + 0.0064431x_1^3x_2 - 0.14066x_1^3 + 0.48519x_1^2x_2^2\\
	          &+ 0.18081x_1^2x_2 + 0.31882x_1^2 + 0.045636x_1x_2^3 - 0.030792x_1x_2^2\\
			  & + 0.0463x_1x_2 + 0.022898x_1 + 0.3829x_2^4 + 0.24085x_2^3\\
			  &+ 0.48187x_2^2 + 0.10909x_2 + 0.37734 \\

    u_3 =  &  0.5497x_1^4 - 0.035471x_1^3x_2 + 0.073809x_1^3 + 0.66023x_1^2x_2^2\\
	       & - 0.085302x_1^2x_2 + 0.34888x_1^2 - 0.020016x_1x_2^3 + 0.55526x_1x_2^2\\
		   &+ 0.032773x_1x_2 - 0.10637x_1 + 0.81332x_2^4 - 0.055596x_2^3\\
		   & + 0.49761x_2^2 + 0.25765x_2 + 0.93038  \\
    \psi_1(\theta)  = & \psi_2(\theta)=-4\theta+2\theta^2,
    \end{array}$ }
     where $\delta$,$u_1$,$u_2$, $u_3 -\chi-u_1g_0$, $-\LIDER_f(\chi)+\psi_1(\chi)$, $-\varphi-u_2g_0$, $-\LIDER_f(\varphi)+\psi_2(\varphi)+\delta \chi$, $\varphi-u_2g_1$ are positive polynomials. \qed
\end{example}

\begin{example}
	\label{ex:7}
Consider an HS with two modes in Fig.~6, in which
 the CDSs  at $q_1$ and $q_2$ are respectively $\dot{\xx} = \fb_1(\xx)$ and $\dot{\xx} = \fb_2(\xx)$, where
\begin{equation*}
 ~\fb_1(\xx) =
	\begin{cases}
		x_2 \\
		- x_1 - x_3\\
		x_1 + ( 2x_2 +3x_3 ) ( 1 + x_3^2 ),
	\end{cases} \hspace*{-1cm}
	\fb_2(\xx) =
	\begin{cases}
		x_2 \\
		- x_1 - x_3\\
		- x_1 - 2x_2 - 3x_3,
	\end{cases}
\end{equation*}
 $\Xi_{q_1} = \{\xx \in \RR^3 \mid x_1^2 + x_2^2 +x_3^2 \leq 0.01\}$, $\Xi_{q_2}=\emptyset$,
  $D_{q_1} = x_1^2 +0.01x_2^2 + 0.01x_3^2 \leq 1.01$,  $D_{q_2} = x_1^2+x_2^2 +x_3^2 \geq 0.03 \wedge x_1^2 \leq{5.1}^2$,
   $g_1 = 0.99 \leq x_1^2 + 0.01 x_2^2 + 0.01 x_3^2 \leq 1.01$, and
$g_2 = 0.03 \leq x_1^2 +x_2^2 +x_3^2 \leq 0.05$. All resets are identity.

The proof obligation is to verify $| x_1 | \leq 3.2$ at $q_2$. To the end, we synthesize barrier certificates at
each mode first (see the appendix), then we need to verify the following  five conditions :\\[2mm]
$\left\{ \begin{array}{l}
  c_1=-\varphi_2-u_{23}g_{11}-u_{24}g_{12}\geq 0,\\
  c_2=-\chi_2 - u_{21}g_{11}-u_{22}g_{12} \geq 0,\\
  c_3=-\LIDER_f(\chi_2)-0.2\chi_2+\chi_2^2-u_{41}D_2-u_{41}D_{21} \geq 0,\\
  c_4=-\LIDER_f(\varphi_2)-0.2\varphi_2+\varphi_2^2-\delta_2\chi_2-u_{51}D_2 -u_{52}D_{21} \geq 0,\\
  c_5=\varphi_2 -U_2-0.00001\geq 0,
   \end{array} \right.$ \\[2mm]
   by {\sdp}. In which,
$g_{11}=1.01-x_1^2-0.01x_2^2-0.01x_3^2$, $g_{12}= x_1^2+0.01x_2^2+0.01x_3^2-0.99$,
$D_2=x_1^2+x_2^2+x_3^2-0.03$,
$D_{21}=26.01-x_1^2$,
$U_2=x_1^2-10.24$,  and
$u_{21},u_{22}$, $u_{23}$, $u_{24}$, $u_{41}$, $u_{42}$, $u_{51}$, $u_{52}$, $\delta_2$ are \textbf{SOS} synthesized in the first step.
 \qed

\begin{figure}
  \begin{center}
    \setlength{\unitlength}{0.6cm}
    \begin{picture}(10,5)(-5,-2.5)
      \allinethickness{1pt}
      \put(-3,0){\circle{2.5}}
      \put(3,0){\circle{2.5}}
      \spline(-1.9,0.55)(0,1.5)(1.9,0.55)
      \spline(-1.9,-0.55)(0,-1.5)(1.9,-0.55)
      \drawline(1.6,0.55)(1.9,0.55)(1.7,0.76)
      \drawline(-1.6,-0.55)(-1.9,-0.55)(-1.7,-0.76)
      \put(-3.1,0){\makebox(0,0)[l]{$q_1$}}
      \put(2.8,0){\makebox(0,0)[l]{$q_2$}}
      \put(-0.2,1.5){\makebox(0,0)[l]{$g_1$}}
      \put(-0.2,-1.5){\makebox(0,0)[l]{$g_2$}}
    \end{picture}
    \caption{An HS with two modes }
  \end{center} \label{fig:a}
\end{figure}
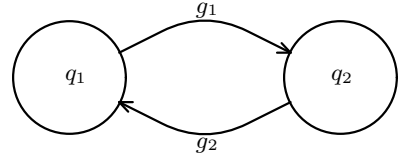
\end{example}

 All the experimental results of all examples given in this paper can be summarised as in Table~\ref{tab:1},
 in which the label $\times$ means that the corresponding method can not obtain a barrier certificate.
All the results listed were computed on a 64-bit Intel(R) Core(TM) i5 CPU 650 @ 3.20GHz with 4GB RAM memory and Ubuntu 12.04 GNU/Linux.

By comparing with the approach reported in \cite{kong2013} (see Table~\ref{tab:1}), our approach can synthesize more barrier certificates,
in particular, with lower degree, but our approach takes more time. However, our approach is still very efficient, typically,
symbolic checking can make our approach to avoid unsoundness because of the error due to
 numeric computation in {\sdp}.

 \begin{center}
	\begin{table}
		\begin{tabular}{|c|c|c|c|c|c|} \hline
       &  \multicolumn{2}{c|}{Exp. cond.} & \multicolumn{3}{c|}{Our method} \\ \cline{2-6}
       &   Degree & Time(s) & Degree & \multicolumn{2}{c|}{Time(s)} \\ \cline{5-6}
       &          &      &        &  {\tiny Synthesis} & {\tiny Symb. checking}\\ \hline
 E.g. \ref{ex:4}& $\times$&  $\times$& 8 & 36.023 & 10.766             \\ \hline
 E.g. \ref{ex:31}& 8&  1.132& 6 & 2.717 & 0.226             \\ \hline
  E.g. \ref{ex:5} & 6 &  1.516& 4 & 4.658 & 0.180             \\ \hline
   E.g. \ref{ex:7} & 4 &  1.387& 2 & 4.260 & 20.472             \\ \hline
 \end{tabular}
\caption{Experimental data. \label{tab:1}}
  \end{table}
\end{center}

\oomit{\scriptsize \begin{center}
	\begin{table}
		\begin{tabular}{ | p {2cm} |  p {1.1cm} | p{1.1cm}|p {1.1cm} |p {1.1cm}|  }
	\hline
	{\tiny Examples} &  {\tiny Example }    &  {\tiny Example \ref{ex:31}} &{\tiny Example \ref{ex:5}}   &{\tiny Example \ref{ex:7}}    \\ \hline
	{\tiny \ec\ (degree)}                &  $\times$   &   8    &  6    &  4     \\ \hline
	{\tiny Our method\ (degree)}         &  8          & 6      &  4    &  2     \\ \hline
	{\tiny \ec\ computing time(s)}       &  $\times$   & 1.132  &1.516  & 1.387  \\ \hline
	{\tiny Numerical computing time(s)}  &  36.023     & 2.717  &4.658  & 4.260  \\ \hline
	{\tiny Symbolic checking time(s)}    & 10.776      & 0.226  & 0.180 & 20.472  \\
	\hline
\end{tabular}
\caption{Experimental data. \label{tab:1}}
  \end{table}
\end{center} }

\section{Concluding Remarks}\label{sec:con}
To summarize, the contributions of this paper include: \vspace*{-.3cm}
 \begin{itemize}
  \item  Relaxation of 
     the conditions of barrier certificate in a general way,
       so that one can utilize weaker conditions flexibly to
       synthesize various kinds of  barrier certificates with more expressiveness,
       which gives more opportunities to verify the considered system. \vspace*{-.1cm}
 \item A method to combining two functions
        together to form a combined barrier certificate in order to prove a safety property under consideration,
        whereas  neither of them can be used as a barrier certificate separately. \vspace*{-.1cm}
 \item   An approach to synthesizing certificates
     according to the general relaxed conditions by semi-definite programming. In particular, we discussed
      how to apply symbolic checking to guarantee the soundness of our approach caused by the error
      of numeric computation in {\sdp}. \vspace*{-.1cm}
 \item Experimental results demonstrating that our approach can indeed discover more certificates and
 give more opportunities to verify an HS under consideration.
 \end{itemize} \vspace*{-.3cm}

 For future work, we plan to combine more than two functions to form a combined barrier certificate. In particular, we are interested in finding more functions $\psi$ satisfying condition (\ref{cond:4}) and establishing a library for them.
 In addition, it is interesting to investigate how to recover the error caused by the numeric computation in {\sdp} by some
 symbolic computation techniques.

\bibliographystyle{abbrv}
\bibliography{references}  

\newpage

\appendix
\section{The details of Examples}

The polynomials synthesized in Example  \ref{ex:4} are: \\
	{\scriptsize
	$ \varphi=  0.030317x_1^8 + 6.9115e-05x_1^7x_2 - 3.6889e-05x_1^7 + 0.090347x_1^6x_2^2 - 0.11095x_1^6x_2 - 0.75683x_1^6
		  - 9.0598e-05x_1^5x_2^3 - 0.00017438x_1^5x_2^2 - 5.4845e-05x_1^5x_2 + 7.291e-05x_1^5 - 0.30715x_1^4x_2^4 + 1.0445
		    x_1^4x_2^3 - 1.5458x_1^4x_2^2 + 0.57141x_1^4x_2 - 0.26344x_1^4 - 6.3369e-05x_1^3x_2^5 + 0.00010503x_1^3x_2^4
			  + 0.00038237x_1^3x_2^3 - 0.00036159x_1^3x_2^2 - 0.00010184x_1^3x_2 + 9.2214e-05x_1^3 + 0.03383x_1^2x_2^6 + 0.33103
			    x_1^2x_2^5 - 2.9864x_1^2x_2^4 + 2.0938x_1^2x_2^3 - 0.12636x_1^2x_2^2 + 0.79519x_1^2x_2 - 0.62237x_1^2 + 1.4962e-05x_1
				  x_2^7 - 0.00014241x_1x_2^6 + 0.00048485x_1x_2^5 - 0.00065416x_1x_2^4 + 0.00014521x_1x_2^3 + 0.00040002x_1x_2^2
				    - 0.00031516x_1x_2 + 6.343e-05x_1 - 0.0043261x_2^8 + 0.05803x_2^7 - 0.29525x_2^6 + 0.80728x_2^5 - 1.2538x_2^4
					  + 1.2862x_2^3 - 0.76567x_2^2 + 0.29172x_2 - 0.072688,$\\
	$\chi=   9.8484x_1^6 + 0.001271x_1^5x_2 + 13.4422x_1^4x_2^2 - 31.2496x_1^4x_2 - 85.8767x_1^4 - 0.0031705x_1^3x_2^2 - 0.012227
	  x_1^3x_2 - 0.0042103x_1^3 + 5.396x_1^2x_2^4 - 28.4976x_1^2x_2^3 - 46.3212x_1^2x_2^2 + 87.5486x_1^2x_2 - 44.1755x_1^2
	    + 0.0049683x_1x_2^2 - 0.0058767x_1x_2 - 0.0020784x_1 + 0.46783x_2^6 - 4.0071x_2^5 + 6.1875x_2^4 + 37.296x_2^3 - 100
		  x_2^2 + 2.0932x_2 - 12.8904,$\\
	$ \delta= 0.014034x_1^8 + 3.0608e-06x_1^7x_2 - 5.813e-06x_1^7 + 0.0021473x_1^6x_2^2 - 0.013483x_1^6x_2 - 0.0064165x_1^6 + 2.5531e-06x_1^5x_2^3 + 2.7689e-05x_1^5x_2^2 - 1.0371e-05x_1^5x_2 - 2.9253e-06x_1^5 + 0.02322x_1^4x_2^4 - 0.008259x_1^4x_2^3 + 0.0095319x_1^4x_2^2 + 0.014437x_1^4x_2 + 0.02625x_1^4 + 1.126e-05x_1^3x_2^5 - 3.9758e-05x_1^3x_2^4 + 3.4426e-05x_1^3x_2^3 - 1.2647e-05x_1^3x_2^2 + 8.5785e-06x_1^3x_2 - 5.7495e-07x_1^3 + 0.00051658x_1^2x_2^6 - 0.010421x_1^2x_2^5 + 0.032928x_1^2x_2^4 - 0.041062x_1^2x_2^3 + 0.030059x_1^2x_2^2 - 0.030394x_1^2x_2 + 0.013862x_1^2 + 5.2631e-07x_1x_2^7 - 4.7302e-06x_1x_2^6 + 1.3673e-05x_1x_2^5 - 1.3654e-05x_1x_2^4 + 8.3569e-07x_1x_2^3 - 1.6543e-06x_1x_2^2 + 1.3272e-05x_1x_2 - 8.3958e-06x_1 + 0.00013121x_2^8 - 0.0015345x_2^7 + 0.0076214x_2^6 - 0.019749x_2^5 + 0.02836x_2^4 - 0.023961x_2^3 + 0.01575x_2^2 - 0.010837x_2 + 0.0044092 ,$\\
	$u_1= 33.1703x_1^4 - 0.0080558x_1^3x_2 - 0.014014x_1^3 + 31.8846x_1^2x_2^2 - 22.7751x_1^2x_2 + 33.5594x_1^2 + 0.002164x_1
	  x_2^3 - 0.0059715x_1x_2^2 - 0.037073x_1x_2 + 0.020061x_1 + 10.479x_2^4 + 6.0815x_2^3 + 19.5851x_2^2 - 18.8795x_2
	    + 24.5699,$\\
 	$ u_2= 0.579x_1^8 + 7.0204e-06x_1^7x_2 - 1.8771e-05x_1^7 + 0.61572x_1^6x_2^2 - 0.43594x_1^6x_2 + 0.39633x_1^6 + 4.0635e-06x_1^5x_2^3 + 5.4444e-06x_1^5x_2^2 - 9.0679e-06x_1^5x_2 + 4.1779e-05x_1^5 + 0.5972x_1^4x_2^4 - 0.446x_1^4x_2^3 + 0.8667x_1^4x_2^2 - 0.48811x_1^4x_2 + 0.57967x_1^4 + 2.0738e-06x_1^3x_2^5 + 4.4963e-06x_1^3x_2^4 - 3.9037e-06x_1^3x_2^3 - 1.5008e-05x_1^3x_2^2 - 6.0762e-05x_1^3x_2 + 3.4303e-05x_1^3 + 0.42761x_1^2x_2^6 - 0.20453x_1^2x_2^5 + 0.45199x_1^2x_2^4 - 0.55762x_1^2x_2^3 + 0.80255x_1^2x_2^2 - 0.18571x_1^2x_2 + 0.36852x_1^2 + 5.1943e-06x_1x_2^7 - 5.229e-06x_1x_2^6 + 4.0646e-06x_1x_2^5 - 8.2704e-06x_1x_2^4 + 1.2324e-05x_1x_2^3 - 3.1593e-06x_1x_2^2 - 8.7493e-06x_1x_2 + 3.0456e-06x_1 + 0.18043x_2^8 + 0.1527x_2^7 + 0.11373x_2^6 + 0.090147x_2^5 + 0.40667x_2^4 - 0.26137x_2^3 + 0.68588x_2^2 - 0.38649x_2 + 0.50807,$\\
  $ u_3=0.82691x_1^8 + 6.8463e-06x_1^7x_2 + 7.824e-06x_1^7 + 0.66339x_1^6x_2^2 + 0.69976x_1^6x_2 + 0.71112x_1^6 + 2.8381e-06x_1^5x_2^3 + 2.1678e-05x_1^5x_2^2 - 2.269e-05x_1^5x_2 - 2.7155e-05x_1^5 + 0.67426x_1^4x_2^4 + 0.31337x_1^4x_2^3 + 1.0011x_1^4x_2^2 + 0.41117x_1^4x_2 + 0.94169x_1^4 + 5.177e-06x_1^3x_2^5 + 2.4687e-05x_1^3x_2^4 + 4.9193e-05x_1^3x_2^3 - 8.5264e-05x_1^3x_2^2 - 9.261e-05x_1^3x_2 - 0.00018356x_1^3 + 0.55287x_1^2x_2^6 + 0.26734x_1^2x_2^5 + 0.43798x_1^2x_2^4 - 0.42762x_1^2x_2^3 + 0.90269x_1^2x_2^2 + 0.47337x_1^2x_2 + 0.73082x_1^2 + 6.4342e-06x_1x_2^7 + 1.7535e-05x_1x_2^6 + 3.191e-05x_1x_2^5 + 2.9182e-05x_1x_2^4 + 7.677e-05x_1x_2^3 + 2.0006e-05x_1x_2^2 - 3.1684e-05x_1x_2 - 6.6486e-06x_1 + 0.45107x_2^8 - 0.15576x_2^7 + 0.12208x_2^6 - 0.29031x_2^5 + 0.39853x_2^4 - 0.57126x_2^3 + 0.29347x_2^2 - 0.33244x_2 + 0.43254 ,$\\
   $\psi_i(\theta)= \psi_2(\theta)=-4\theta+2\theta^2.$ }

The polynomials synthesized in Example \ref{ex:31} are : \\
{\scriptsize
$\varphi=  9.8484x_1^6 + 0.001271x_1^5x_2 + 13.4422x_1^4x_2^2 - 31.2496x_1^4x_2 - 85.8767x_1^4 - 0.0031705x_1^3x_2^2 - 0.012227x_1^3x_2 - 0.0042103x_1^3 + 5.396x_1^2x_2^4 - 28.4976x_1^2x_2^3 - 46.3212x_1^2x_2^2 + 87.5486x_1^2x_2 - 44.1755x_1^2 + 0.0049683x_1x_2^2 - 0.0058767x_1x_2 - 0.0020784x_1 + 0.46783x_2^6 - 4.0071x_2^5 + 6.1875x_2^4 + 37.296x_2^3 - 100x_2^2 + 2.0932x_2 - 12.8904,$\\
	$\chi=0,\delta=0,u1=0,$ \\
	$u_2=  9.8484x_1^6 + 0.001271x_1^5x_2 + 13.4422x_1^4x_2^2 - 31.2496x_1^4x_2 - 85.8767x_1^4 - 0.0031705x_1^3x_2^2 - 0.012227x_1^3x_2 - 0.0042103x_1^3 + 5.396x_1^2x_2^4 - 28.49    76x_1^2x_2^3 - 46.3212x_1^2x_2^2 + 87.5486x_1^2x_2 - 44.1755x_1^2 + 0.0049683x_1x_2^2 - 0.0058767x_1x_2 - 0.0020784x_1 + 0.46783x_2^6 - 4.0071x_2^5 + 6.1875x_2^4 + 37.296x_2^3 - 1    00x_2^2 + 2.0932x_2 - 12.8904,$\\
	$u_3=9.8484x_1^6 + 0.001271x_1^5x_2 + 13.4422x_1^4x_2^2 - 31.2496x_1^4x_2 - 85.8767x_1^4 - 0.0031705x_1^3x_2^2 - 0.012227x_1^3x_2 - 0.0042103x_1^3 + 5.396x_1^2x_2^4 - 28.49    76x_1^2x_2^3 - 46.3212x_1^2x_2^2 + 87.5486x_1^2x_2 - 44.1755x_1^2 + 0.0049683x_1x_2^2 - 0.0058767x_1x_2 - 0.0020784x_1 + 0.46783x_2^6 - 4.0071x_2^5 + 6.1875x_2^4 + 37.296x_2^3 - 1    00x_2^2 + 2.0932x_2 - 12.8904, $ \\
$\psi_1=0,\psi_2(\theta)=-4\theta+1.5\theta^2 $}.

The polynomials synthesized in Example \ref{ex:7} are:
  {\scriptsize
$\varphi_2 = 1.6165x_1^2 - 0.20569x_1x_2 + 0.19824e-1x_1x_3 + 0.95436e-5x_1 + 0.54446e-1x_2^2 +
0.69996e-3x_2x_3 - 0.16916e-6x_2 + 0.9101e-1x_3^2 + 0.1511e-7x_3 - 9.6424$\\
$\chi_2 = 0.89818e-1x_1^2 - 0.82739e-1x_1x_2 + 0.21192e-1x_1x_3 - 0.15224e-8x_1 +
0.54928e-2x_2^2 + 0.84123e-2x_2x_3 + 0.1277e-8x_2 + 0.35173e-1x_3^2 + 0.27238e-9x_3 - 5.3973$\\
$\delta_2 = 5.5914x_1^2 - 0.21067x_1x_2 - 0.24733e-1x_1x_3 + 0.87702e-5x_1 + 0.20573x_2^2 -
0.52174e-1x_2x_3 + 0.28769e-6x_2 + 0.22449x_3^2 + 0.87144e-7x_3 + 0.29484$\\
$u_{21} = 1.5356x_1^2 + 0.13731e-1x_1x_2 - 0.19249e-2x_1x_3 - 0.10079e-6x_1 + 0.66295x_2^2 -
0.64549e-1x_2x_3 - 0.63485e-7x_2 + 0.39611x_3^2 - 0.66953e-8x_3 + 2.6867$\\
$u_{22} = 0.73288x_1^2 - 0.22775e-2x_1x_2 + 0.27401e-2x_1x_3 - 0.51154e-7x_1 + 0.59472x_2^2 -
0.55279e-1x_2x_3 - 0.48206e-7x_2 + 0.34978x_3^2 - 0.74061e-8x_3 + 0.60632$\\
$u_{23} = 2.0821x_1^2 + 0.40593e-1x_1x_2 - 0.50855e-2x_1x_3 - 0.8427e-4x_1 + 0.61146x_2^2 -
0.90046e-2x_2x_3 - 0.83808e-5x_2 + 0.14389x_3^2 - 0.1148e-5x_3 + 4.5124$\\
$u_{24} = 1.0004x_1^2 + 0.1131e-1x_1x_2 - 0.22779e-2x_1x_3 - 0.288e-4x_1 + 0.517x_2^2 - 0.80914e-2x_2x_3 -
0.11074e-4x_2 + 0.83205e-1x_3^2 - 0.82264e-6x_3 + 0.70099$\\
$u_{41} = 0.43056e-3x_1^2 - 0.29796e-4x_1x_2 + 0.10489e-3x_1x_3 + 0.59287e-11x_1 + 3.8141e-6x_2^2 +
0.95752e-5x_2x_3 + 0.26518e-12x_2 + 0.39903e-4x_3^2 + 0.51833e-12x_3 + 0.36e-2$\\
$u_{42} = 0.56936e-2x_1^2 + 0.53069e-2x_1x_2 + 0.35737e-2x_1x_3 - 0.75779e-10x_1 + 0.20039e-2x_2^2 +
0.26891e-2x_2x_3 + 0.29818e-10x_2 + 0.16505e-2x_3^2 - 0.16159e-10x_3 + 0.52902$\\
$u_{51} = 0.28447e-1x_1^2 - 0.28324e-2x_1x_2 - 0.37952e-3x_1x_3 + 0.125e-6x_1 + 0.52784e-3x_2^2 -
0.86183e-4x_2x_3 - 0.63026e-8x_2 + 0.88611e-4x_3^2 - 0.86257e-9x_3 + 0.11143e-1$\\
$u_{52} = 0.12129x_1^2 - 0.13405e-1x_1x_2 - 0.15008e-2x_1x_3 - 0.82285e-6x_1 + 0.47845e-2x_2^2 -
0.73624e-3x_2x_3 - 0.20348e-6x_2 + 0.52816e-3x_3^2 + 0.21178e-7x_3 + 3.5079$
}

  \end{document}